\documentclass[11pt]{amsart}

 \usepackage[margin=1in]{geometry}

\newcommand{\private}[1]{}
\usepackage{amssymb, amsmath, amscd, amsthm, color, epsfig, url, tikz, graphicx, hyperref, booktabs, pifont}

\usepackage[all]{xy}          % for xy-pic pictures
\xyoption{dvips}              % for xy-pic pictures
\newcommand{\xmark}{\textcolor{red}{\ding{55}}}

\makeatletter
\renewcommand\l@subsection{\@tocline{2}{0pt}{2pc}{5pc}{}}
\makeatother

\tikzset{
  mynode/.style={circle, draw, minimum size=7mm, inner sep=0pt},
  myarrow/.style={->, >=Stealth, thick},
  mylabel/.style={midway, above}
}

\usetikzlibrary{arrows.meta, positioning}

\theoremstyle{plain}
\newtheorem{thm}{Theorem}[section]
\newtheorem{prop}[thm]{Proposition}

\theoremstyle{definition}
\newtheorem{defin}[thm]{Definition}
\newtheorem{example}[thm]{Example}
\newtheorem{def/ex}[thm]{Definition/Example}

\theoremstyle{remark}

\makeatletter
\@namedef{subjclassname@2020}{%
  \textup{2020} Mathematics Subject Classification}
\makeatother

\begin{document}

%%%%%%%%%%%%%%%%%%%%%%%%%%%%%%%%%%%%%%%%%%%%%%%%%%%%%%%%%%%%%%%%%%%%%%%%%%%%%%%%%%%%%%%%

\title[BallotRank]{BallotRank: A Condorcet Completion Method for Graphs}

%%%%%%%%%%%%%%%%%%%%%%%%%%%%%%%%%%%%%%%%%%%%%%%%%%%%%%%%%%%%%%%%%%%%%%%%%%%%%%%%%%%%%%%%

\author{Jason  Douglas Todd}
\address{Division of Social Sciences, Duke Kunshan University, No. 8 Duke Avenue, Kunshan, China, 215316}
\email{jason.todd@dukekunshan.edu.cn}
\urladdr{https://sites.google.com/view/jasondtodd}

\author{Ismar Voli\'c}
\address{Department of Mathematics \& Statistics, Wellesley College, 106 Central Street, Wellesley, MA 02481}
\email{ivolic@wellesley.edu}
\urladdr{ivolic.wellesley.edu}

%%%%%%%%%%%%%%%%%%%%%%%%%%%%%%%%%%%%%%%%%%%%%%%%%%%%%%%%%%%%%%%%%%%%%%%%%%%%%%%%%%%%%%%%

%%%%%%%%%%%%%%%%%%%%%%%%%%%%%%%%%%%%%%%%%%%%%%%%%%%%%%%%%%%%%%%%%%%%%%%%%%%%%%%%%%%%%%%%%%%%%%%%%%%%%%%%

\begin{abstract}
We introduce BallotRank, a ranked preference aggregation method derived from a modified PageRank algorithm. It is a Condorcet-consistent method without damping, and empirical examination of nearly 2,000 ranked choice elections and over 20,000 internet polls confirms that BallotRank always identifies the Condorcet winner at conventional values of the damping parameter. We also prove that the method satisfies many of the same social choice criteria as other well-known Condorcet completion methods, but it has the  advantage of being a natural social welfare function that provides a full ranking of the candidates.
\end{abstract}


%%%%%%%%%%%%%%%%%%%%%%%%%%%%%%%%%%%%%%%%%%%%%%%%%%%%%%%%%%%%%%%%%%%%%%%%%%%%%%%%%%%%%%%%%%%%%%%%%%%%%%%%

\maketitle 

\tableofcontents

\parskip=6pt
\parindent=0cm

%%%%%%%%%%%%%%%%%%%%%%%%%%%%%%%%%%%%%%%%%%%%%%%%%%%%%%%%%%%%%%%%%%%%%%%%%%%%%%%%%%%%%%%%
\section{Introduction}\label{S:Intro}
%%%%%%%%%%%%%%%%%%%%%%%%%%%%%%%%%%%%%%%%%%%%%%%%%%%%%%%%%%%%%%%%%%%%%%%%%%%%%%%%%%%%%%%%
Ranked choice voting is having a moment. In addition to the intense media coverage of the recent New York City mayoral election, where primaries were held under ranked choice rules, there are also new electoral reforms in Hawaii, the District of Columbia, and elsewhere. Although to the lay public instant runoff (IRV) and its multi-winner variant, single-transferable vote (STV), are synonymous with ranked choice voting, these methods represent but the tip of the iceberg to social choice theorists. Here, older systems like the Borda count and Copeland/Llull method compete for attention with more modern inventions such as minimax and ranked pairs. We add to this list by introducing a new ranked voting method, which we call \emph{BallotRank}.

BallotRank builds upon the equivalence between a preference profile and a graph first noted in \cite{mcgarvey1953theorem} and extended by \cite{BJNR:ConvergenceVoting} to produce the so-called ``Condorcet graph.''\footnote{See \cite{Volic2024} for an interesting discussion of the equivalence.} By modifying this graph 
% through the use of margins and the addition of self-loops 
and then applying a variant of the PageRank algorithm \cite{bp1998cn}, BallotRank becomes a natural social welfare function. If one elects to eliminate the PageRank damping parameter by setting it to 1, BallotRank is a Condorcet-consistent method satisfying many well-known social choice criteria. If instead the damping parameter is set to conventional values such as 0.85, it still correctly identifies the Condorcet winner in every single one of the more than 20,000 elections we examine. 

BallotRank is also a more natural social welfare function than some of the other Condorcet completion methods like minimax, ranked pairs, and Schulze/beatpath. The calculation of the rankings, including identifying the winner, requires input from the entire system. The power method of calculating BallotRank iteratively updates the rankings, but it needs the full information from the previous step to do so. Other methods are primarily designed to find the Condorcet winner or break Condorcet cycles, and the fact that they can be used to produce a ranking is an afterthought. BallotRank's final ranking, on the other hand, is canonical and intrinsic to the method. 

Several recent works have proposed PageRank- or Markov chain–based methods for aggregating pairwise comparisons. 
Negahban et al. \cite{NOS:IterativeRanking, NOS:RankCentrality} propose \emph{rank centrality}, an iterative spectral algorithm that recovers latent item scores from noisy pairwise comparisons under the Bradley–Terry–Luce model. These methods interpret rankings as estimators of qualities and focus on consistency, sample complexity, and optimality guarantees. From a social choice theory viewpoint, Altman and Tennenholtz \cite{AT:PageRankAxioms} provide an axiomatic characterization of PageRank as a ranking system on directed graphs, establishing a representation theorem based on graph-theoretic axioms tailored to web-ranking contexts. Brandt and Fischer \cite{BF:PageRankTournament} build on this and analyze PageRank as a weak tournament solution. 

The above papers are far apart in intent and spirit from the constructions in this paper. Research closer to ours was done by Bana et al. \cite{BJNR:ConvergenceVoting}, who introduce \emph{convergence voting}, a PageRank-inspired social choice function defined on pairwise comparison graphs and interpreted as the outcome of an iterative negotiation process converging to a stationary distribution. Their method is explicitly framed as a consensus rule lying between Condorcet and Borda methods and does not aim to satisfy Condorcet consistency. We will provide more details about the difference between convergence voting and BallotRank in Section \ref{sec:comparison}.

The paper is structured as follows. In Section \ref{sec:pagerank}, we first present the intuition and mathematics behind the PageRank algorithm before modifying it to account for the possibility of cycles. In Section \ref{sec:toyelection}, we work through a toy example to fix ideas. Section \ref{sec:empirical} then examines BallotRank empirically, applying it to online polling and real-world legislative elections and paying particular attention to the problem of Condorcet winners, which BallotRank routinely identifies. Section \ref{sec:comparison} compares the empirical performance of BallotRank to minimax, ranked pairs, and Schulze methods -- standard and familiar methods -- as well as the newer convergence voting mentioned above. As we will demonstrate, BallotRank can behave differently from these methods, and we regard this as an advantage owing to the fact that BallotRank uses more information from the profile to determine the ranking than these other methods. We switch gears in Section \ref{sec:socialchoice} to consider BallotRank against a standard set of social choice criteria including monotonicity, Pareto, and independence of irrelevant alternatives. We conclude that its behavior is largely in line with the standard Condorcet completion methods mentioned above, but there are also appreciable differences. Finally, Section \ref{sec:conclusion} offers a summary of the results and possible directions for further investigation.

%%%%%%%%%%%%%%%%%%%%%%%%%%%%%%%%%%%%%%%%%%%%%%%%%%%%%%%%%%%%%%%%%%%%%%%%%%%%%%%%%%%%%%%%
\section{From PageRank to BallotRank} \label{sec:pagerank}
%%%%%%%%%%%%%%%%%%%%%%%%%%%%%%%%%%%%%%%%%%%%%%%%%%%%%%%%%%%%%%%%%%%%%%%%%%%%%%%%%%%%%%%%
This section presents PageRank and BallotRank intuitively, formalizes both, and translates them into the language of linear algebra. Along the way, we prove that under certain conditions, BallotRank is a Condorcet-consistent method.

\subsection{Social choice under PageRank and BallotRank}
The intuition behind PageRank begins with the common topological interpretation of a random walk along a network, but first we must be clear about the network itself. We first introduce an unweighted graph and add weights, then sketch out the PageRank algorithm, and finally discuss the necessary addition of self-loops.

\subsubsection{Basic network}
In our setting, $n$ nodes represent candidates (more generally, choices) and edges run from losers to winners.\footnote{This differs from the graph theoretic depiction of \cite{mcgarvey1953theorem}, in which edges run from winners to losers.} We will come back to the issue of determining the losers and winners of head-to-head contests in Section \ref{sec:basicalgorithm}.

We assume the network is fully connected in the sense that it is complete if regarded as an undirected graph. With edge orientations, the in-degree and out-degree vary depending upon candidate popularity. A Condorcet winner, by definition, beats all other candidates and hence has an out-degree of zero and an in-degree equal to $n-1$; a Condorcet loser, conversely, has an out-degree of $n-1$ and an in-degree of zero.

\subsubsection{Adding weights}
The first complication involves moving from this unweighted network, in which edges define a binary social preference, to a weighted network, with weights defined by pairwise margins of victory. Thus if candidate $a$ is preferred to candidate $b$ by 363 voters and $b$ is preferred to $a$ by only 300, the network would feature an edge from node $b$ to node $a$ with weight 63. In matrix form, this network graph is an $n \times n$ adjacency matrix whose $A_{ij}$ entry equals $j$'s margin of victory if $j$ is relatively more preferred and equals 0 if $i$ is more preferred.\footnote{This also differs from what \cite{BJNR:ConvergenceVoting} call a ``Condorcet graph,'' in which edges run in both directions with weights defined by pairwise votes, rather than a single edge weighted by the difference between them.}

Starting from a randomly selected candidate/node, a random walk along a network in this setting would involve traveling along one of the outgoing edges to a more preferred candidate, with probabilities proportional to the edge weights. Yet the presence of Condorcet winners or losers presents difficulties: the random walker would never leave the Condorcet winner node and never reach (or return to) a Condorcet loser node. Because a similar problem exists in the original PageRank domain of the web, where some pages feature no outgoing links and others are never linked to, the algorithm presents a solution in the form of probabilistic teleportation. 

\subsubsection{Adding the damping parameter}
The key innovation of Brin and Page was introducing the damping parameter, $d$. With probability $d<1$,\footnote{Since \cite{bp1998cn}, the damping parameter has conventionally been set to 0.85, implying a teleportation probability of 0.15.}  the walker travels along an outgoing edge as described above, while with probability $(1-d)>0$, she teleports to a random node where each has an equal likelihood of serving as the destination. Now all candidate nodes are reachable from all other candidates, and in the limit of an infinitely long random walk, each candidate has a particular probability of hosting the walker. This probability is the candidate's PageRank. 

A second way to intuit what PageRank does requires reimagining the nodes as buckets for holding liquid and the edges as pipes for transporting flows between them. Edge weight now corresponds to the volume of flow an edge is capable of supporting. Measuring PageRank on this network is equivalent to injecting a fixed quantity of liquid, normalized to one, and noting where and in what quantities it pools. In this analogy, the damping factor might correspond to the constant evaporation of a particular proportion of the liquid, which then precipitates uniformly across the network. 

The liquid itself is usefully thought of as something like a social consensus or the pooled votes of the electorate, and PageRank measures to which candidates they flow and in what volumes. Should a Condorcet winner exist, every drop of liquid that does not evaporate due to the damping parameter will pool there. On the other hand, all liquid will flow away from Condorcet losers as the damping parameter approaches one and the evaporation process stalls. If the damping parameter is set to one and evaporation ceases, the PageRank algorithm will unerringly choose the Condorcet winner where it exists.

\subsubsection{Adding self-loops}
Yet even with the addition of the damping parameter, problems remain in the social choice setting. Imagine a five-candidate race with a three-candidate top cycle 
% (majorities of voters prefer candidate $a$ to $b$, $b$ to $c$, and $c$ to $a$) 
and two clear losers
% , meaning that voters prefer $a$, $b$, and $c$ to these two
. All three candidates in the Condorcet cycle are socially preferred to the losing duo, yet each in-cycle candidate can only defeat one of the other at the top. Applying PageRank to the network as described above would yield suitably low scores to the losers and differentiate between them, with the Condorcet loser ranking last. But inside the Condorcet cycle, all three candidates would split the remaining PageRank equally between them. Because each loses only once, the probability of transferring PageRank to the socially preferred candidate equals one and the relative magnitudes of their pairwise margins are irrelevant. 

The solution is to add self-loops to the network, replacing the matrix diagonal of zeroes with a vector of positive numbers which permit the algorithm to distinguish between more popular and less popular nodes in the Condorcet cycle. By adding a self-loop, we allow an in-cycle candidate to retain some PageRank rather than redistributing it all to the socially preferred candidate. Intuitively, self-loop weights should increase in proportion to candidate popularity. 

We propose weighting candidate $a$'s self-loop by candidate $a$'s share of the total win margin, defined as $a$'s summed weighted in-degree divided by the network's total degree.\footnote{At an earlier stage, we also experimented with self-loop weights defined as $\max(\text{total off-diagonal in-degree} - \text{total off-diagonal out-degree}, 0)$. This method did not always identify the Condorcet winner, so we abandoned it.} Thus when a top-cycle candidate barely loses to a fellow in-cycle candidate yet soundly defeats those outside the cycle, she will likely retain considerable PageRank; but when a top-cycle candidate is soundly defeated by another top-cycler, she will lose relatively more PageRank to her socially preferred peer.

\subsection{Formalizing PageRank and BallotRank}
We now turn to formalizing the intuition. We first define PageRank on an unweighted network without teleportation, then add the damping parameter, edge weights, and self-loops, thereby defining BallotRank. 

\subsubsection{The basic algorithm}\label{sec:basicalgorithm}
Consider a set of $n$ \emph{candidates} (or \emph{outcomes} or \emph{choices}) $\mathcal C=\{1, 2, \dots, n\}$ running in an election. Denote by $\mathcal V$ the (finite) set of \emph{voters} (or \emph{agents}). Suppose that for each pair of candidates $(a,b)\in \mathcal C\times\mathcal C$, each voter has a preference for one over the other. The preference might be inferred from a voter's ranked ballot or by the voter considering all pairs of candidates and recording a preference in each case. The ordering is strict and total, so, in particular, it is transitive and there are no ties. The collection of all ranked ballots or pairwise preferences across all the voters is called a \emph{profile}. 
% \footnote{Ismar: Do we consider the case where the ranking or pairwise preferences might not be complete (some candidates might be left off the ranked ballot or the voter may not express a preference between some pairs of candidates)? Or do we for simplicity assume that this is always the case? All the real elections we consider have incomplete ballots, so how do we account for that?}

If more voters prefer candidate $a$ to candidate $b$, we say  that $a$ \emph{beat} $b$, or that $b$ \emph{lost to} $a$. We will denote this situation by $a\succ b$. If there is a tie, then there is no edge between the candidates.\footnote{These ties -- or the admission of incomplete ballots -- could in theory lead to disconnected graphs. Yet the only way of obtaining a fully disconnected node is for a candidate to be equally preferred to every other candidate, highly unlikely in practice. Similarly, the only way that incomplete ballots could yield a disconnected graph would be for disjoint subsets of voters to rank disjoint subsets of candidates, also unheard of in real-world settings. For these reasons, we assume that graphs are always connected.}
% \footnote{Ismar: Say something later about how this might cause a disconnected graph. Also, if we allow incomplete ballots, that could be another reason why the graph is disconnected, namely if disjoint subsets of voters only rank disjoint subsets of candidates (this gets back to my previous footnote about complete rankings). We should acknowledge this but then define this away and just say we'll assume the graph is always connected (since disconnectedness wouldn't really happen in practice?)} 
Note that it is possible for a profile to contain \emph{Condorcet cycles}, namely pairwise preferences that yield $a_1\succ a_2 \succ \cdots\succ a_k\succ a_1$ for some $k\geq 3$.

Let $W(a)$ be the set of candidates losing to $a$ in such pairwise comparisons, and $L(b)$ the total number of pairwise losses by $b$. Translated into the standard PageRank setup that starts with a directed graph, candidates correspond to the nodes and edges point from $b$ to $a$ if $b$ lost to $a$. The out-degree of $b$ is precisely $L(b)$.

If candidate $a$ beats all other candidates (i.e., if $W(a)=\mathcal C\setminus \{a\}$), she is termed the \emph{Condorcet winner}. Similarly, the \emph{Condorcet loser} is a candidate $a$ that loses all pairwise contests (i.e., $W(a)=\emptyset$). Condorcet winners or losers may or may not exist, but when they do, they are unique.

Ignoring the damping factor for now, define the \emph{PageRank $PR(a)$ of candidate $a$} to be 
\begin{equation} \label{eq:undamped}
    PR(a) = \sum_{b \in W(a)} \frac{PR(b)}{L(b)}.
\end{equation}

In prose, the PageRank of any given candidate is simply the sum of the PageRanks of those candidates that lose to it in head-to-head comparisons divided by their total pairwise losses. This is a recurrence that starts with an initial distribution of $1/n$ for all $a$. Thus $PR(a)$ should be understood as the limiting or stationary solution of the recurrence.

It is a well-known feature of PageRank that if a node is a unique sink with no outgoing edges and $n-1$ incoming edges, then PageRank will assign this node a value of 1 in the stationary distribution while assigning 0 to all other nodes.\footnote{With multiple sinks, only one will receive a PageRank of 1, but the identity of the node will depend upon the initial distribution of PageRanks and the graph structure. But that situation would not occur in the social choice setting and can be safely ignored.} The initial ``mass'' of the system eventually concentrates at the sink, or, in the analogy above, all the liquid pools at this node. In our setting, when a Condorcet winner exists, it will be precisely this unique sink. In other words, ignoring the damping parameter (or, equivalently, setting it to 1) means that PageRank is a Condorcet-consistent method.

\subsubsection{Adding the damping parameter}
The damping parameter, $0<d<1$, ensures that the recurrence does not terminate at the sink and allocate all PageRank to the Condorcet winner: 
\begin{equation} \label{eq:damped}
    PR(a) = \frac{1 - d}{n} + d \sum_{b \in W(a)} \frac{PR(b)}{L(b)}.
\end{equation}

Equation (\ref{eq:damped}) represents a random walk along the edges of the graph with a probability $1-d$ of teleportation to a random node. It also defines a Markov chain, a stochastic process in which the probability of the next state depends only on the current state. PageRank can thus be thought of as the stationary probability distribution of this Markov chain.

%With probability $d$, PageRank is defined as in  \eqref{eq:undamped}; with probability $1-d$, PageRank is equal to 1 spread equally across all candidates. In this formulation, the total amount of PageRank sums to unity, enabling interpretation as a stationary probability distribution.

\subsubsection{Adding weights}

In addition to pairwise victories, we also want to consider the margins of victory via edge weights. Let $M(a,b)$, $a\neq b$, be the margin by which $a$ beats $b$. If there is a tie or if $a$ loses to $b$, then $M(a,b)=0$.  Let the \emph{total loss margin} of $b$ be
\begin{equation}\label{eq:totalloss}
\widetilde{TL}(b)=\sum_{a\in\mathcal C, \, a\neq b} M(a,b).
\end{equation}

%Define $TL(b)$, the \emph{total loss margin} of $b$ to be
% $$
% TL(b)=\sum_{a\in \mathcal C, \, a\neq b}M(a,b).
% $$

Then the weighted version of \eqref{eq:damped} is
\begin{equation} \label{eq:weighteddamped}
    \widetilde{PR}(a) = \frac{1 - d}{n} + d \sum_{b\in \mathcal C, \, a\neq b} \frac{M(a,b)}{\widetilde{TL}(b)}\widetilde{PR}(b).
\end{equation}
The interpretation of the denominator is that, when $b$ transfers some of its PageRank to a socially preferred alternative, it does so in proportion to how badly that loss compares to its other losses.

% Then the weighted version of \eqref{eq:damped} is
% \begin{equation} \label{eq:weighteddamped}
%      PR(a) = \frac{1 - d}{n} + d \sum_{b\in \mathcal C, \, a\neq b} \frac{M(a,b)}{{TW}}PR(b).
%  \end{equation}
% Thus when $b$ redistributes its PageRank, it does not do so equally; rather, it splits it proportionally to the strength of its losses. If $b$ loses by a large margin to $a$, it passes on a larger proportion of its PageRank to $a$ than it would if it had barely lost.

\subsubsection{Adding self-loops} Although we now have a standard weighted PageRank setup, this will always rank every candidate in a Condorcet cycle equally. What we want, instead, is a procedure which can distinguish between members of a cycle on the basis of their performances against those in and out of the cycle. To achieve this, we add weighted self-loops to the graph, essentially allowing a candidate to retain a share of her PageRank proportional to the strength of her wins. 

Define $TW(a)$, the \emph{total win margin} of $a$, to be
$$
TW(a)=\sum_{b\in \mathcal C, \, a\neq b}M(a,b).
$$
Also define the \emph{total win margin} of the profile to be
$$
TW=\sum_{a\in \mathcal C}TW(a).
$$
To define the diagonals of the matrix, let
\begin{equation}\label{eq:diagonal}
M(a,a)=\frac{TW(a)}{TW}.
\end{equation}
With the diagonal suitably defined, we can modify our definition of total loss margin $\widetilde{TL}(b)$ from Equation (\ref{eq:totalloss}) to include the $M(a,a)$ term:
\begin{equation}\label{eq:totallossself-loops}
TL(b)=\sum_{a\in\mathcal C} M(a,b).
\end{equation}

% Define the \emph{total win margin} of $b$ to be
% $$
% TM(b)= \sum_{a\in \mathcal C}M(a,b), \ \ a\neq b.
% $$
% We will treat $TM(b)$ as a self-loop at $b$.

% \begin{defin}\label{def:BallotRank}
%     Define \emph{BallotRank} of candidate $a$ as
% \begin{equation}\label{eq:diagweighted}
%   BR(a)\;=\;\frac{1-d}{n}\;+\;d\sum_{b\in\mathcal C}
%   \frac{\,M(a,b)\;+\;\delta_{ab}\;TM(b)\,}{\,TL(b)+TM(b)\,}\;BR(b),
% \end{equation}
% where $\delta_{ab}$ equals $1$ if $a=b$ and is $0$ otherwise.
% \end{defin}

\begin{defin}\label{def:BallotRank}
    Define the \emph{BallotRank} of candidate $a$ as
\begin{equation}\label{eq:diagweighted}
  BR(a)\;=\;\frac{1-d}{n}\;+\;d\sum_{b\in\mathcal C}
  \frac{M(a,b)}{TL(b)}\;BR(b).
\end{equation}
\end{defin}

% {\bf Notation:}\footnote{Ismar: We may not need this remark if we just stick to $d=1$.} When we need to be explicit about the choice of the damping parameter, we will denote $BR(a)$ by $BR_d(a)$. The undamped BallotRank can thus denoted by  $BR_1(a)$.

Thus defined, BallotRank lets each candidate retain a portion of her BallotRank proportional to the intensity of her wins. As before, candidate $b$ redistributes a share of her BallotRank to every candidate that defeats her, proportional to the corresponding loss margin $M(-,b)$, but now, with self-loops, she also keeps a share for herself proportional to her total win margin. In the context of a Condorcet cycle, this modification ensures that BallotRank does not circulate endlessly within the cycle but begins to pool with those candidates posting better victory margins.
% \footnote{Ismar: We should say here why we're doing this; what it does for us. Basically we're trying to resolve loops, so how does this modification does that?}

% \begin{rem} Ismar: We might want to say here that we tried other things, like bottleneck loops and edges with total vote numbers rather than margins, and say why we didn't think those approaches were as good as this one. Maybe here can also say that yet another variant is in \cite{BJNR:ConvergenceVoting}, but we didn't think this was natural either because of some result it gave on some election (I don't remember which at the moment). Also, I think that might not always pick the Condorcet winner, and that's probably because they don't look at margins, just raw votes with edges going in both directions. Can we say something about why passing to margins always guarantees that we get the Condorcet winner? \end{rem}

\subsection{Linear algebra interpretation of BallotRank}

For computational purposes,  it is helpful to formulate BallotRank in the language of linear algebra. 
Define the nonnegative $n\times n$ matrix  $\boldsymbol{M}=(M(a,b))$, $1\leq a,b\leq n$, by the values of $M(a,b)$ as above, that is, the win margins of $a$ over $b$ when $a\neq b$ and by Equation \eqref{eq:diagonal} when $a=b$.
% \[
% A_{ab} =
% \begin{cases}
% M(a,b), & a\neq b,\\[4pt]
% TM(b), & a=b.
% \end{cases}
% \]
% So, in the $b$th column, off–diagonal entries store loss margins of candidate $b$ and the diagonal entry has $b$'s total win margin.

% We want to normalize the columns; adding the values in column $b$ (the total ``outflow weight'' in column $b$) gives
% \[
% S(b) \;=\; \sum_{a=1}^n A_{ab} \;=\; TL(b) + TM(b).
% \]
Finally, define the column‑stochastic matrix $\boldsymbol{\ell}=(\ell(a,b))$
 by
\begin{equation}
\ell(a,b) \;=\; \frac{M(a,b)}{TL(b)}.
% \ell(a,b) \;=\; \frac{A_{ab}}{S(b)}
% \;=\; \frac{M(a,b)+\mathbf \delta_{ab}TM(b)}{TL(b)+TM(b)}.
\end{equation}
By construction, each column of $\boldsymbol{\ell}$ adds to $1$ and $\ell(a,b)=0$ whenever $b$ beats $a$ (and $a\neq b$).

We will occasionally have need for the BallotRank matrix without the diagonal self-loops. We will denote that matrix by 
$\tilde{\boldsymbol{\ell}}=(\tilde{\ell}(a,b))$, corresponding to the weighted PageRank expression from Equation \eqref{eq:weighteddamped}.

If we collect the BallotRanks of all candidates into a vector
\[
    \mathbf{BR} = \begin{pmatrix}
        BR(1)\\
        BR(2)\\
        \vdots\\
        BR(n)
    \end{pmatrix},
\]
then the \emph{BallotRank vector} is the unique solution of the fixed-point equation
\begin{equation}\label{eq:power}
   \mathbf{BR} \;=\; \frac{1-d}{n}\,\mathbf{1} \;+\; d\,\boldsymbol{\ell}\,\mathbf{BR},
\end{equation}
where $\mathbf{1}$ is the column vector of 1's. This simply says that $\mathbf{BR}$ is the stationary distribution of the random walk governed by $\boldsymbol{\ell}$ with teleportation probability $1-d$.  
Equivalently, BallotRank is the dominant right eigenvector of a rescaled adjacency matrix corresponding to the eigenvalue 1, and the solution is given by
\begin{equation} \label{eq:solve}
    \mathbf{BR} = (\mathbf{I} - d \boldsymbol{\ell})^{-1}\cdot \frac{1-d}{n}\mathbf{1}.
\end{equation}

In practice, $\mathbf{BR}$ can be approximated to an arbitrary precision by the power method, or simple iteration of  \eqref{eq:power}. Alternatively, popular network analysis packages in \texttt{R} and \texttt{Python} contain implementations of PageRank with user-specified $d$ (but defaulting to $d=0.85$) that solve Equation \eqref{eq:solve}. However, as $d$ approaches unity, the eigenvector-based approach fails (because $\mathbf{I} - d \boldsymbol{\ell}$ is singular and cannot be inverted), while the power method remains viable for all values of $d$.

In summary, starting with a ranked profile or pairwise preferences, BallotRank produces a full ranking of the candidates. At $d=1$, BallotRank will always identify the Condorcet winner if one exists, making it a \emph{Condorcet completion} or an \emph{extended Condorcet} method.\footnote{It is worth noting, however, that when $d=1$ and a Condorcet winner exists, BallotRank fails to produce a full ranking, instead allocating 1 to the Condorcet winner and 0 to all other candidates.} In theory, BallotRank might not identify the Condorcet winner at $d\geq 0.85$, but as we will see in Section \ref{sec:empirical}, this does not seem to happen in practice. Other familiar Condorcet completion methods include minimax, ranked pairs, and Schulze/beatpath. They also pick the Condorcet winner when one exists, and differ only in the case of a Condorcet cycle. While these methods can also be redefined so as to produce candidate rankings, BallotRank naturally yields a ranking because that was the original intent behind PageRank. We compare BallotRank with these other methods in Section \ref{sec:comparison}.
% \footnote{Ismar: Maybe we should emphasize this more? BallotRank is somehow more ``natural'' than the others? Maybe also give more background in intro about Condorcet methods and how they all just try to resolve the issue of cycles in different ways. We should also mention Arrow in intro.}

\section{Working through a toy election}\label{sec:toyelection}
\label{sec:example}
Suppose we have a set of candidates $\mathcal C=\{a, b, c, d\}$ and 12 voters are asked to rank them. One voter's ranking is $a\succ b\succ c\succ d$, seven voters submitted $a\succ b\succ d\succ c$, and the last four ranked $b\succ a\succ c\succ d$. The following matrix encodes the margins:
$$    
 \begin{pmatrix}
0 & 4 & 12 & 12 \\ 
0 & 0 & 12 & 12\\
0 & 0 & 0 & 0 \\
0 & 0 & 2 & 0 
\end{pmatrix}
$$
To interpret the matrix, consider the 12 in the second row and third column. This entry, $M(b,c)$, means that $b$ beat $c$ by 12 votes (because the electorate unanimously preferred $b$ to $c$). Translating this to a graph, we have

\begin{center}
\begin{tikzpicture}[node distance=3cm]
  % Nodes in a square
  \node[mynode] (a) {$a$};
  \node[mynode, right of=a] (b) {$b$};
  \node[mynode, below of=a] (d) {$d$};
  \node[mynode, below of=b] (c) {$c$};

  % Directed edges with adjusted labels
  \draw[myarrow] (b) -- (a) node[midway, above] {$4$};
  \draw[myarrow] (c) -- (b) node[pos=0.5, right, yshift=4pt] {$12$}; % moved up
  \draw[myarrow] (c) -- (a) node[pos=0.8, right, yshift=4pt] {$12$};
  \draw[myarrow] (d) -- (a) node[pos=0.5, left, yshift=4pt] {$12$}; % moved up
  \draw[myarrow] (d) -- (b) node[pos=0.8, left, yshift=4pt] {$12$};
  \draw[myarrow] (c) -- (d) node[midway, above] {$2$};
\end{tikzpicture}
\end{center}
The graph makes it clear that $a$ is the Condorcet winner. In the matrix, a Condorcet winner exists when every non-diagonal entry in a row is non-zero. It is also evident that $c$ is the Condorcet loser; in matrix form, the Condorcet loser is marked by a row containing only zeros.

Calculating total wins, or row sums, we have $$ TW(a)=28, \ \ TW(b)=24, \ \ TW(c)=0, \ \ TW(d)=2, $$
and $TW=28+24+0+2=54$. The diagonal entries are thus
$$ M(a,a)=\frac{28}{54}=\frac{14}{27}, \ \  M(b,b)=\frac{24}{54}=\frac{4}{9}, \ \  M(c,c)=\frac{0}{54}=0, \ \  M(d,d)=\frac{2}{54}=\frac{1}{27}. $$
% To compute the diagonal entries, we need total win and loss margins. Total win margins $TM$ are calculated by adding across rows and total loss margins $TL$ are obtained by adding across columns:
% \begin{align*}
%     TM(a)& =28, \ \ \ 
%     TL(a)=0; \\
%     TM(b)& =24, \ \ \ \, 
%     TL(b)=4; \\
%     TM(c)& =0, \ \ \ \ \ 
%     TL(c)=26; \\
%     TM(d)& =2, \ \ \ \ \ 
%     TL(d)=24. \\
% \end{align*}
% To make the BallotRank matrix, the denominator in column $i$ is $TL(i)+TM(i)$. We divide each entry in $i$th column by that value and give the diagonal entry the value which makes the column add to 1, so $1-(\sum\text{off-diagonal entries in column $i$})$.  We thus get
% $$  
% \ell=    
% \begin{pmatrix}
% 1 & \frac{4}{28} &  \frac{12}{26}  &  \frac{12}{26}  \\ 
% 0 & \frac{24}{28} & \frac{12}{26} &  \frac{12}{26}\\
% 0 & 0 & 0  &  0\\
% 0 & 0 &  \frac{2}{26}  &   \frac{2}{26}
% \end{pmatrix}
% $$

To make the BallotRank matrix, $\boldsymbol{\ell}$, we normalize each column, which finally yields
$$  
\boldsymbol{\ell}=    
\begin{pmatrix}
1 & \frac{9}{10} &  \frac{6}{13}  &  \frac{324}{649}  \\[4pt] 
0 & \frac{1}{10} & \frac{6}{13} &  \frac{324}{649}\\[2pt]
0 & 0 & 0  &  0\\[2pt]
0 & 0 &  \frac{1}{13}  &   \frac{1}{649}
\end{pmatrix}
$$
The graph now looks like

\begin{center}
\begin{tikzpicture}[node distance=3cm]
  % Nodes in a square
  \node[mynode] (a) {$a$};
  \node[mynode, right of=a] (b) {$b$};
  \node[mynode, below of=a] (d) {$d$};
  \node[mynode, below of=b] (c) {$c$};

  % Directed edges with adjusted labels
  \draw[myarrow] (b) -- (a) node[midway, above] {$\tfrac{9}{10}$};
  \draw[myarrow] (c) -- (b) node[midway, right] {$\tfrac{6}{13}$};
  \draw[myarrow] (c) -- (a) node[pos=0.8, right, yshift=4pt] {$\tfrac{6}{13}$};    % moved up along arrow
  \draw[myarrow] (d) -- (a) node[midway, left] {$\tfrac{324}{649}$};
  \draw[myarrow] (d) -- (b) node[pos=0.8, left, yshift=4pt] {$\tfrac{324}{649}$}; % moved up along arrow
  \draw[myarrow] (c) -- (d) node[midway, above] {$\tfrac{1}{13}$};

  % Self-loops (rounder, correct orientation)
  \draw[myarrow] (a) edge[loop left,  min distance=10mm, in=150, out=210] node {$1$} (a);
  \draw[myarrow] (b) edge[loop right, min distance=10mm, in=30, out=330] node {$\tfrac{1}{10}$} (b);
  \draw[myarrow] (c) edge[loop right, min distance=10mm, in=330, out=30] node {$0$} (c);
  \draw[myarrow] (d) edge[loop left,  min distance=10mm, in=210, out=150] node {$\tfrac{1}{649}$} (d);
\end{tikzpicture}
\end{center}

Calculating BallotRank with the conventional $d=0.85$ and without damping ($d=1$) produces
\begin{center}
\begin{tabular}{l|cccc}
\toprule
               & $a$ & $b$ & $c$ &  $d$ \\
\midrule
$d=0.85$ & 0.8469 & 0.0756 & 0.0375 & 0.0400 \\
$d=1.00$ & 1.0000 & 0.0000 & 0.0000 & 0.0000 \\
\bottomrule
\end{tabular}
\end{center}
When $d=0.85$, BallotRank produces the ranking $a\succ b\succ d\succ c$. The result for $d=1$ illustrates that, without damping, the Condorcet winner $a$ absorbs all BallotRank in the system. If one only cares about identifying a single winner or determining whether a Condorcet winner exists, then setting $d$ to 1 is useful. But if one desires a full ranking in the presence of a Condorcet winner, then it is necessary to set $d$ to something strictly less than 1.

\section{Empirical analyses}\label{sec:empirical}
%%%%%%%%%%%%%%%%%%%%%%%%%%%%%%%%%%%%%%%%%%%%%%%%%%%%%%%%%%%%%%%%%%%%%%%%%%%%%%%%%%%%%%%%
This section first examines four datasets from diverse decision-making settings: an online polling service; American local, state, and federal elections; Scottish local councils using STV; and Australian state legislative elections. We then delve deeper into four notorious American elections, two without Condorcet winners and two in which IRV failed to crown the Condorcet winner.

%%%%%%%%%%%%%%%%%%%%%%%%%%%%%%%%%%%%%%%%%%%%%%%%%%%%%%%%%%%%%%%%%%%%%%%%%%%%%%%%%%%%%%%%
\subsection{BallotRank and Condorcet winners in real elections}
%%%%%%%%%%%%%%%%%%%%%%%%%%%%%%%%%%%%%%%%%%%%%%%%%%%%%%%%%%%%%%%%%%%%%%%%%%%%%%%%%%%%%%%%
We begin our empirical analyses with a dataset from the Condorcet Internet Voting Service (CIVS), a free online platform for ranked choice elections, containing 22,309 usable polls, of which 73.62\% yield a Condorcet winner.\footnote{These data, last updated December 15, 2024, are freely available at \url{https://civs1.civs.us/data-releases.html}.} Figure \ref{fig:all-data} presents this subset of the CIVS data in pink, highlighting that they substantially differ from real-world legislative elections. In the subset of CIVS elections that have a Condorcet winner, the mean number of choices is 9, the mean number of ballots cast is 33, and on average there were nearly seven times as many votes as choices. Given the unusual diversity of poll parameters represented, we consider the CIVS data to be a difficult first test for BallotRank. Applying BallotRank to the 16,425 elections with a Condorcet winner, we test values of the damping parameter in $\{0.5, 0.85, 0.999, 1\}$. Under these settings, and despite the oddities of the CIVS data, BallotRank successfully identifies the Condorcet winner in every single poll. 

\begin{figure}
    \centering
    \includegraphics[width=\linewidth]{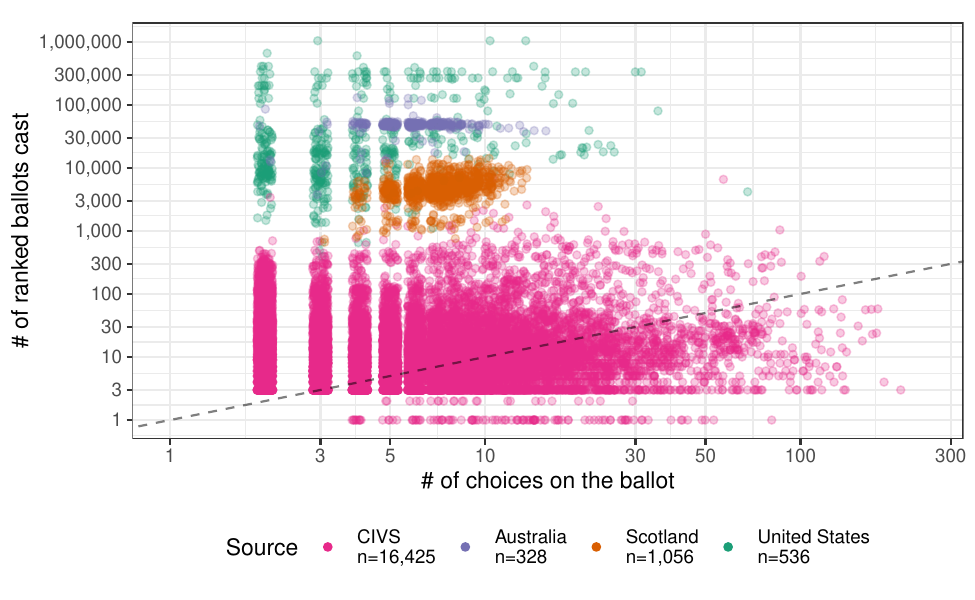}
    \caption{18,345 elections with a Condorcet winner exhibit considerable variation in choice set and electorate sizes; BallotRank correctly pegs the Condorcet winner every time. Data are plotted on a log-log scale.}
     \label{fig:all-data}
\end{figure}

Turning to elections with real political consequences, we begin with a set of 538 ranked choice elections from across the United States, all but two of which (99.63\%) have a Condorcet winner.\footnote{We discuss the elections without Condorcet winners, as well as two Condorcet elections in which instant-runoff voting failed to select the Condorcet winner, in the next subsection.} Sourced from the FairVote Dataverse repository \cite{FairVote}, these elections include 28 officer elections from the American Psychological Association from 1998 to 2021 and 510 elections from across the country to local, state, and federal offices from 2004 to 2024. The mean number of candidates is 5, the mean number of voters is 70,794, and the average ratio is 17,384 voters per candidate. Figure \ref{fig:all-data} displays the elections in this dataset in green, highlighting the reduction in choices and greatly increased electorate relative to the online polls analyzed above. We again apply BallotRank for values of $d \in \{0.5, 0.85, 0.999, 1\}$ and in all cases the Condorcet winner is correctly identified.

Another dataset covers 1,070 local Scottish council elections, 1,056 (98.69\%) of which feature a Condorcet winner. These data, from David McCune and Alexander Graham-Squire \cite{MGS:Scotland}, include almost every council election from 2012, 2017, and 2022.\footnote{As discussed in their paper, they in turn source the 2012 and 2017 data from Andrew Teale's \textit{Local Elections Archive Project} \cite{Teale:LEAP}.} In these mostly multiwinner single-transferable vote elections, the average number of candidates is 7 and the mean number of voters is 5,125. Figure \ref{fig:all-data} presents these data in orange, where the average election features 719 voters per candidate. Running BallotRank on $d \in \{0.5, 0.85, 0.999, 1\}$, we again find that the Condorcet winner is correctly identified every time.

The final real-world dataset, courtesy of Nicholas O. Stephanopoulos \cite{S:FindingCondorcet}, comes from Australia and appears in purple in Figure \ref{fig:all-data}. These 328 elections include mayoral elections in New South Wales from 2021, New South Wales state legislative elections from 2015--2023, and Victoria state legislative elections from 2018--2022. Every election in this dataset has a Condorcet winner. The average numbers of candidates and voters are 6 and 48,475, respectively; the average ratio between them is 8,588 voters per candidate. Application of BallotRank at parameter values $d \in \{0.5, 0.85, 0.999, 1\}$ yet again produces a perfect correspondence between Condorcet winners and BallotRank winners.

%%%%%%%%%%%%%%%%%%%%%%%%%%%%%%%%%%%%%%%%%%%%%%%%%%%%%%%%%%%%%%%%%%%%%%%%%%%%%%%%%%%%%%%%
\subsection{A closer look at four U.S. elections}
%%%%%%%%%%%%%%%%%%%%%%%%%%%%%%%%%%%%%%%%%%%%%%%%%%%%%%%%%%%%%%%%%%%%%%%%%%%%%%%%%%%%%%%%
% \ifn{I don't see the actual BallotRanks for these elections, maybe we should add them in?}
On March 3, 2009, voters in the city of Burlington, Vermont went to the polls to elect a mayor under a recently instituted IRV system.\footnote{See \cite{Ellenberg2015} for a lengthier treatment of the Burlington case.} The five candidates on the ballot included Progressive Party incumbent Bob Kiss, Democrat Andy Montroll, Green James Simpson, independent Dan Smith, and Republican Kurt Wright. Under the conventional first-past-the-post (FPTP) system, Wright would have won; under the extant IRV system, Kiss won instead. But, notoriously, Andy Montroll was the Condorcet winner and was not elected. As the top-left panel of Figure \ref{fig:odd-data} reveals, however, BallotRank at \textit{any} value of the damping parameter greater than zero successfully crowns Montroll as the winner. 

The panel presents the results from applying BallotRank to the Burlington mayoral preference profile for all values of $d \in \{0, 0.01, 0.02, \dots, 0.98, 0.99, 1\}$, as well as the final ranking and BallotRank values when $d=1$. Note that when a Condorcet winner exists and $d$ is set to 1, the probability of teleportation falls to zero, all that matters is the network structure, and the Condorcet winner acts as a sink, causing the BallotRank values to collapse to $\{0, 1\}$. In other words, at $d=1$ in the presence of a Condorcet winner, BallotRank values function as an indicator variable for the Condorcet winner and cease to provide a complete ranking of the candidates. Such a ranking is available, however, at all values of $d<1$, and the labels in Figure \ref{fig:odd-data} reflect this ranking: BallotRank places the IRV winner (Kiss) ahead of the FPTP winner (Wright).

\begin{figure}
    \centering
    \includegraphics[width=\linewidth]{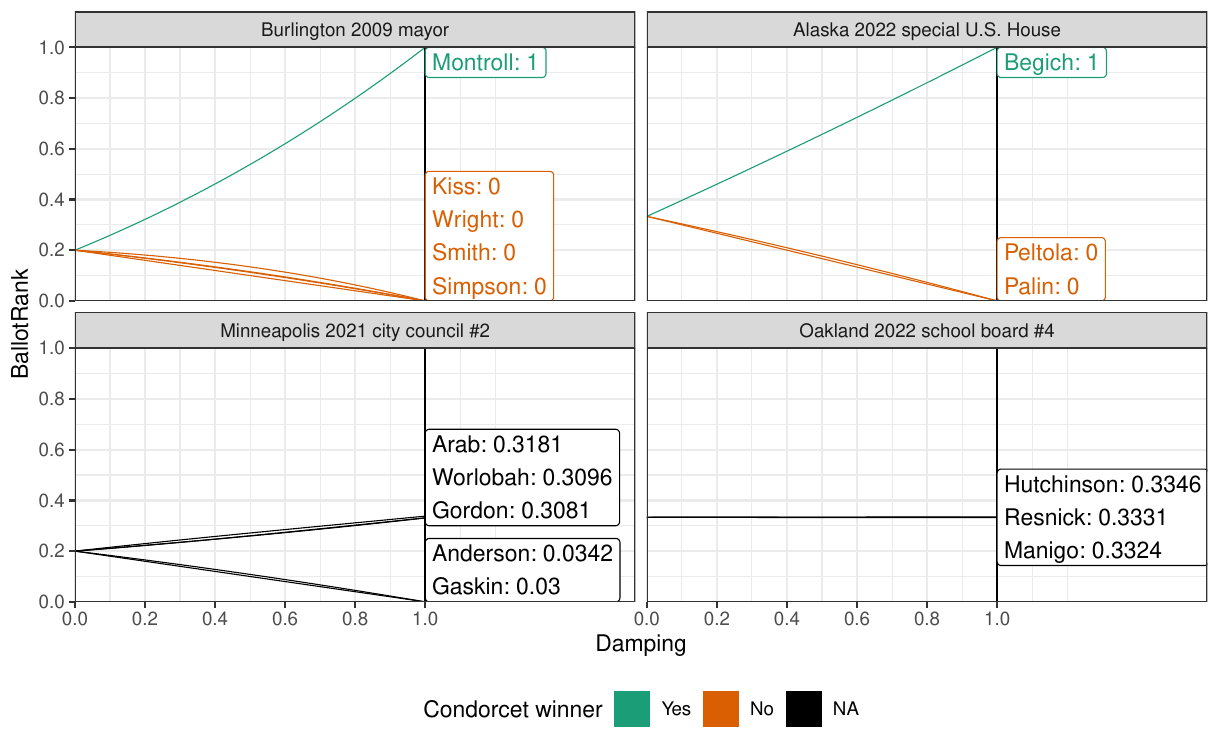}
    \caption{Four case studies demonstrate BallotRank's superiority to IRV, its behavior across the range of $d$, its convenient production of a complete ranking, and its ability to rank candidates within cycles.}
    \label{fig:odd-data}
\end{figure}

The special election for Alaska's at-large congressional district, held on August 16, 2022, is similarly infamous amongst election aficionados for failing to elect the Condorcet winner when one existed. Held to fill a seat vacated upon the death of incumbent Republican Don Young, this election pitted Democrat Mary Peltola against Republicans Sarah Palin and Nick Begich III in Alaska's first experiment with IRV. When the results were announced two weeks later, Democrat Peltola was the upset winner and Begich was the first eliminated.\footnote{Peltola would also have won under FPTP rules.} Yet despite earning the fewest first-preference votes, Republican Begich was in fact the Condorcet winner. But as the top-right panel of Figure \ref{fig:odd-data} reveals, BallotRank again successfully crowns Begich the winner for all $d>0$. Also noteworthy, Peltola ranks ahead of Palin for all $0<d<1$. 

Another pair of real-world elections is remarkable in a different way -- in neither case does a Condorcet winner exist. On November 2, 2021, the roughly 10,000 Minneapolis voters of Ward 2 headed to the polls to elect their local city council member. Their choices were Democratic-Farmer-Labor candidates Tom Anderson and Yusra Arab, Republican Guy Gaskin, incumbent Green Cam Gordon, and Democratic Socialist Robin Wonsley Worlobah. Under IRV rules, Anderson and Gaskin were eliminated in round 1, Gordon was nixed in round 2, and Worlobah defeated Arab in the third and final round.\footnote{Worlobah would also have won under FPTP rules.} 

The bottom-left panel of Figure \ref{fig:odd-data} shows that BallotRank distinguishes between two groups of candidates and selects Arab as the winner over Worlobah. The explanation for both facts lies in the structure of the electorate's preference profile. After calculating pairwise differences between each candidate (see Table \ref{tab:Minneapolis}), it becomes clear that a top-cycle exists: Gaskin beats no one; Anderson beats only Gaskin; and Arab, Gordon, and Worlobah each lose only once and only amongst themselves. Thus, Anderson and Gaskin score BallotRank values near zero, with Anderson ever so slightly higher, and the remaining three candidates score just below $1/3$, as they split the bulk of the total BallotRank. 

\begin{table}[t]
\begin{center}
\begin{tabular}{l|rrrrr|rr|l}
\toprule
               & Gordon & Arab & Worlobah &  Gaskin & Anderson & $\sum$ wins & Total & Diagonal \\
\midrule
Gordon     &   0 &  0 & 73 & 4,592 & 2,192 & 6,857 & 24,455 & 0.280 \\
Arab       & 225 &  0 &  0 & 5,079 & 3,239 & 8,543 & 24,455 & 0.349 \\
Worlobah   &   0 & 15 &  0 & 4,283 & 2,017 & 6,315 & 24,455 & 0.258 \\
Gaskin     &   0 &  0 &  0 &     0 &     0 &     0 & 24,455 & 0.00 \\
Anderson   &   0 &  0 &  0 & 2,740 &     0 & 2,740 & 24,455 & 0.112 \\
\midrule
$\sum$ losses & 225 & 15 & 73 & 16,694 & 7,448 & & & \\
\bottomrule
\end{tabular}
\end{center}
\caption{Win margins and initial diagonals (before column normalization) in the 2021 Minneapolis Ward 2 election. Read 73 in the first row as ``Gordon beat Worlobah by 73 votes.''}
    \label{tab:Minneapolis}
\end{table}

\begin{figure}[h]
\begin{center}
\begin{tikzpicture}[>=stealth,scale=1.7]
  % Styles
  \tikzstyle{cand}=[circle, draw, thick, minimum size=8mm, inner sep=0pt]
  \tikzstyle{edge}=[->, thick]

  % Nodes on a pentagon (radius=1.8cm)
  \node[cand] (gordon)   at (90:1.8)   {Go}; % top
  \node[cand] (arab)     at (18:1.8)   {Ar}; % upper right
  \node[cand] (worlobah) at (306:1.8)  {Wo}; % lower right
  \node[cand] (anderson) at (234:1.8)  {An}; % lower left
  \node[cand] (gaskin)   at (162:1.8)  {Ga}; % upper left

  % --- Edges (rounded to 3 d.p.) ---

  % Column: Gordon
  \draw[edge] (gordon) -- node[midway, pos=0.4, xshift=20pt] {0.999} (arab);
  \draw[edge] (gordon) edge[loop above, min distance=8mm, in=120, out=60]
      node[above] {\small 0.001} (gordon);

  % Column: Arab
  \draw[edge] (arab) -- node[right] {0.978} (worlobah);
  \draw[edge] (arab) edge[loop right, min distance=8mm, in=30, out=330]
      node[right] {\small 0.022} (arab);

  % Column: Worlobah
  \draw[edge] (worlobah) -- node[below, rotate=-70] {0.997} (gordon);
  \draw[edge] (worlobah) edge[loop right, min distance=8mm, in=330, out=30]
      node[right] {\small 0.004} (worlobah);

  % Column: Gaskin
  \draw[edge] (gaskin) -- node[midway, pos=0.6, xshift=-20pt]  {0.275} (gordon); 
  \draw[edge] (gaskin) -- node[above] {0.304} (arab);
  \draw[edge] (gaskin) -- node[below, rotate=-40] {0.257} (worlobah);
  \draw[edge] (gaskin) -- node[left]  {0.164} (anderson);
  \draw[edge] (gaskin) edge[loop left, min distance=8mm, in=150, out=210]
      node[left] {\small 0} (gaskin);

  % Column: Anderson
  \draw[edge] (anderson) -- node[below, rotate=70] {0.294} (gordon);
  \draw[edge] (anderson) -- node[below, rotate=40]      {0.435} (arab);
  \draw[edge] (anderson) -- node[below]     {0.271} (worlobah);
  \draw[edge] (anderson) edge[loop left, min distance=8mm, in=210, out=150]
      node[left] {\small $1.5\cdot 10^{-5}$} (anderson);
\end{tikzpicture}
\end{center}
\caption{BallotRank graph for the 2021 Minneapolis Ward 2 election.}
    \label{fig:Minneapolis}
\end{figure}
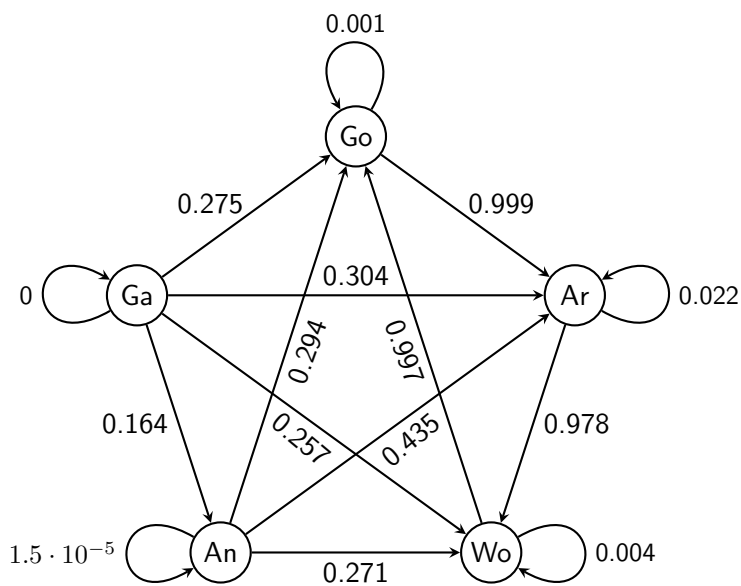

But how does BallotRank determine the relative rankings \textit{within} the top-cycle, and thereby the winner? Here we leverage the pairwise margins, which largely determine the edge weights between nodes, and the total margin of victory, which dictates the self-loop edge weights (see Figure \ref{fig:Minneapolis}). Recall that candidates with larger total margins of victory are initially awarded self-loops with larger edge weights. But to find BallotRank's stationary distribution, we must normalize columns to sum to one. In doing so, large initial diagonals (e.g., Arab's 0.35) and low total \textit{loss} margins (e.g., Arab's 15) translate into larger weights on the self-loops. In effect, candidates are rewarded for winning big and losing small by permitting them to retain more of their BallotRank at each iteration of the power method. 

Arab has the largest total win margin at 8,543, yielding the largest initial diagonal. At the same time, Arab has the smallest total loss margin at 15. In normalizing Arab's column, we divide the diagonal by the new column sum ($\frac{0.35}{0.35+15}$) to obtain a self-loop weight of 0.022. Hers is the largest self-loop weight because her numerator is largest and her denominator smallest. By winning big and losing small, Arab finishes first in BallotRank. Although Gordon barely tops Worlobah in total wins, Worlobah more decisively beats Gordon in total losses; ultimately, Worlobah (the IRV winner) squeaks out a second-place finish ahead of incumbent Gordon. 

The last case study comes from the Oakland Unified School District in California, where on November 8, 2022 voters in District 4 faced three nonpartisan candidates for the local school board: Mike Hutchinson, Pecolia Manigo, and Nick Resnick. In round 1, Resnick had the most first-preference votes and Manigo, with the fewest, was eliminated. In round 2 -- after discovering and correcting an error in the IRV algorithm which had erroneously crowned Resnick the winner -- the Alameda County Registrar of Voters announced Mike Hutchinson had won by 256 votes.\footnote{For more drama, see \url{https://www.nbcbayarea.com/decision-2022/oakland-school-board-race/3116249/}.} 

The bottom-right panel of Figure \ref{fig:odd-data} shows that for all values of $d$, application of BallotRank yields a tight race -- but that Hutchinson wins with Resnick a close second. As the preceding example would suggest, the explanation for the tight bunching of BallotRank values is a classic cycle: Hutchinson beats Resnick by 299 votes, Manigo beats Hutchinson by 46, and Resnick beats Manigo by 557. Here there are only three candidates and hence no higher-order network effects. Effectively, Hutchinson wins under BallotRank because the ratio of his win margin to his loss margin (6.5) is the highest. Resnick (1.86) places second, and Manigo (0.08) ranks last. See Table \ref{tab:Oakland} and Figure \ref{fig:Oakland} for more details.

\begin{table}[h]
    \begin{center}
        \begin{tabular}{l|rrr|rr|l}
            \toprule
            & Hutchinson & Resnick & Manigo & $\sum$ wins & Total & Diagonal \\
            \midrule
            Hutchinson &  0 & 299 &    0 & 299 & 902 & 0.331 \\
            Resnick    &  0 &   0 &  557 & 557 & 902 & 0.618 \\
            Manigo     & 46 &   0 &    0 &  46 & 902 & 0.051 \\
            \midrule
            $\sum$ losses & 46 & 299 & 557  & & & \\
            \bottomrule
        \end{tabular}
    \end{center}
    \caption{Win margins and initial diagonals (before column normalization) in the 2022 Oakland school board District 4 election. Read 46 in the final row as ``Manigo beat Hutchinson by 46 votes.''}
    \label{tab:Oakland}
\end{table}

\begin{figure}[h]
 \begin{center}
\begin{tikzpicture}[>=stealth,scale=1.3]
  % Styles
  \tikzstyle{cand}=[circle, draw, thick, minimum size=7mm, inner sep=0pt]
  \tikzstyle{edge}=[->, thick]

  % Nodes on a triangle (radius=1.4cm)
  \node[cand] (Hu) at (90:1.4)   {Hu};  % Hutchinson - top
  \node[cand] (Re) at (330:1.4)  {Re};  % Resnick - bottom right
  \node[cand] (Ma) at (210:1.4)  {Ma};  % Manigo - bottom left

  % --- Edges (rounded to 3 d.p.) ---

  % From Hu (Hutchinson): loses only to Ma (46), self diag 0.331
  \draw[edge] (Hu) -- node[midway, yshift=2pt, xshift=-16pt] {0.993} (Ma);
  \draw[edge] (Hu) edge[loop above, min distance=7mm, in=120, out=60]
      node[above] {\small 0.007} (Hu);

  % From Re (Resnick): loses only to Hu (299), self diag 0.618
  \draw[edge] (Re) -- node[midway, yshift=2pt, xshift=16pt] {0.998} (Hu);
  \draw[edge] (Re) edge[loop right, min distance=7mm, in=30, out=330]
      node[right] {\small 0.002} (Re);

  % From Ma (Manigo): loses only to Re (557), self diag 0.051
  \draw[edge] (Ma) -- node[below, pos=0.55] {0.999} (Re);
  \draw[edge] (Ma) edge[loop left, min distance=7mm, in=150, out=210]
      node[left] {\small $9.2\cdot 10^{-5}$} (Ma);
\end{tikzpicture}
\end{center}
\caption{BallotRank graph for the 2022 Oakland election.}
    \label{fig:Oakland}
\end{figure}
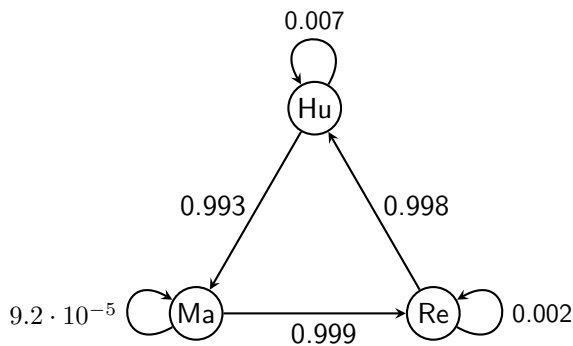

These four case studies demonstrate that at suitably large $d$, BallotRank crowns the Condorcet winner when one exists, identifies a reasonable winner when one does not exist, and produces different results from those obtained under the conventional IRV method. In the next section, we expand the comparison to additional methods with a mixture of real-world elections and contrived examples.

%%%%%%%%%%%%%%%%%%%%%%%%%%%%%%%%%%%%%%%%%%%%%%%%%%%%%%%%%%%%%%%%%%%%%%%%%%%%%%%%%%%%%%%%
\section{Comparison with other methods}\label{sec:comparison}
%%%%%%%%%%%%%%%%%%%%%%%%%%%%%%%%%%%%%%%%%%%%%%%%%%%%%%%%%%%%%%%%%%%%%%%%%%%%%%%%%%%%%%%%
Here we compare BallotRank with three additional Condorcet-consistent methods -- minimax, ranked pairs, Schulze (or beatpath) \cite{simpson1969defining, kramer1977dynamical, tideman1987independence, S:ShultzeMethod} -- and another PageRank-inspired method mentioned in the Introduction called convergence voting \cite{BJNR:ConvergenceVoting}. Because, like BallotRank, the first three methods always select the Condorcet winner when one exists, we examine only the 16 real-world elections without a Condorcet winner, two from the United States and 14 from Scotland. Table \ref{tab:no-condorcet} lists these contests; the winners under BallotRank, minimax, ranked pairs, and Schulze; and an indicator for complete agreement between them. We contrast BallotRank with convergence voting separately; although convergence voting embraces a similar network logic as BallotRank, it is \textit{not} guaranteed to find the Condorcet winner when one exists.

The headline finding here is that although the three Condorcet-completion comparison methods achieve complete agreement (barring the complication of ties), BallotRank frequently differs from the consensus in its choice of winner. Among this sample of 16 elections, BallotRank selected the consensus winner only six times -- and in two of those instances, BallotRank's choice is contingent upon the choice of the damping parameter, $d$. 

Before exploring an example of $d$-dependence, it is worth emphasizing that BallotRank's frequent disagreement with minimax, ranked pairs, and Schulze is a feature, not a bug. Minimax, for example, reduces to a single number for each candidate -- the strength of their worst defeat. Once this number is fixed, all other wins and losses are ignored. Schulze looks at strongest paths, but once the strongest path comparison is fixed, no other information is needed.   BallotRank, on the other hand, iteratively utilizes more information from the preference profile, including second-, third-, and higher-order network effects, as well as candidates' total win and loss margins, to produce its output. We should therefore expect its output to differ from those of other methods, particularly in the presence of preference cycles. 

At the same time, as a social welfare function, BallotRank also produces output that is richer in information. Rather than merely identifying a winner, BallotRank -- as its name implies -- produces a complete ranking of all candidates.\footnote{Unless $d=1$ and a Condorcet winner exists, in which case BallotRank values collapse to 1 for the Condorcet winner and 0 for everyone else.} This may prove fruitful in the case of multiwinner contests, as second or third place finishers are identified without additional rounds of electoral mathematics. Full rankings may also yield useful information for losers, who may use it when estimating the feasibility of running in the next electoral cycle.

\begin{table}[h]
    \centering
    \begin{tabular}{l|l|l|l|l|c}
        \toprule
        Election & BallotRank & Minimax & Ranked Pairs & Schulze & Match \\
        \midrule
        Oakland 2022 \#4 & Hutchinson & Hutchinson & Hutchinson & Hutchinson & $\checkmark$ \\
        Minneapolis 2021 \#2 & Arab & Worlobah & Worlobah & Worlobah & \\
        \midrule
        Argyll Bute 2022 \#2 & McFadzean & Currie & Currie & Currie & \\
        Angus 2012 \#3 & Devine & Middleton & Middleton & Middleton & \\
        Dumgal 2012 \#12 & Grant* & Grant & Grant & Grant & $\checkmark$ \\
        Dumgal 2022 \#3 & Murdoch & Murdoch & Murdoch & Murdoch & $\checkmark$ \\
        Fife 2012 \#16 & Crichton & Sloan & Sloan & Sloan & \\
        Borders 2012 \#9 & Wyse & Wyse & Wyse & Wyse & $\checkmark$ \\
        Aberdeenshire 2022 \#10 & Grant* & Payne & Payne & Payne & \\
        Angus 2022 \#3 & Clark & Devine & Devine & Devine & \\
        Glasgow 2017 \#21 & Kelly & TIE (Kelly) & - & - & \\
        N. Lanarkshire 2017 \#4 & Barclay* & Johnston & Johnston & Johnston & \\
        Renfrewshire 2012 \#7 & Caldwell & Caldwell & Caldwell & Caldwell & $\checkmark$ \\
        W. Dunbartonshire 2022 \#2 & Millar* & Millar & Millar & Millar & $\checkmark$ \\
        Glasgow 2007 Govan & Dornan & Flanagan & Flanagan & Flanagan & \\
        N. Ayrshire 2012 \#3 & Cullinane & TIE (Cullinane) & - & - & \\
        \bottomrule
    \end{tabular}
    \caption{BallotRank presents a stark contrast to three popular Condorcet-consistent methods because it leverages additional information about voters' preferences. A * indicates that BallotRank winners depend upon the value of $d$, with the listed victor winning at $d=1$.}
    \label{tab:no-condorcet}
\end{table}

To understand why BallotRank winners occasionally depend upon the value of $d$ used, we focus on the May 5, 2022 council election in Aberdeenshire's West Garioch (\#10) Ward. Here eight candidates, each from a different party or an independent, competed for three seats under STV rules. A top-cycle exists, with Moray Grant, Sam Payne, and incumbent Hazel Smith defeating the rest of the field while losing once within the cycle. With 1,247 and 1,150 first-preference votes, respectively, Grant and Payne surpassed the quota (1,060) and gained seats in round 1. Votes were then transferred repeatedly as candidates were eliminated, ultimately leaving Smith as the winner of the final seat in round 8. 

Under all three comparison methods, Payne is crowned the victor. BallotRank, however, selects Smith for values of $d \in (0, 0.85]$ and Grant at all higher values -- but never Payne. Why? At lower values of $d$, redistribution effects are of greater importance, while at higher values of $d$ the network structure predominates. In other words, at low $d$ it matters relatively more how a top-cycle candidate fares against out-of-cycle competitors, who receive redistributed BallotRank and then pass it along to top-cycle candidates. But at high $d$ little redistribution occurs, the top-cycle becomes a sink, and what matters most is how a top-cycle candidate fares relative to her within-cycle competitors. 

Considering total win margins in the top-cycle, Smith has 4,712, Grant comes second with 3,979, and Payne last with 3,198. Among the five out-of-cycle candidates, four of them lose by a larger margin to Smith than to either Grant or Payne. Payne, in fact, beats the out-of-cycle candidates by the smallest margin in four out of five cases. So at low $d$, where BallotRank contributions from out-of-cycle candidates matter, it should be no surprise that Smith wins. But at high $d$ what matters are the edge weights within the top-cycle. These, in turn, partially depend upon total loss margins: Grant - 9.23, Smith - 14.27, Payne - 70.19. Ultimately, Grant winds up with the largest self-loop weight (0.025), with Smith's close behind (0.019), and Payne's a distant third (0.003). At high $d$, these weights are determinative -- and Grant wins.
% \footnote{Jason @ Ismar: should I dig in further to understand why Payne wins across the comparison methods? I've thoroughly convinced \textit{myself}, at least, that Payne is a clear third-place finisher, so I fail to see why he would win.}

As a final point of comparison, we turn to convergence voting (CV) \cite{BJNR:ConvergenceVoting}, which begins from a similar premise to that which underlies BallotRank -- preference profiles can be represented as graphs from which PageRank can produce a single distribution of weights. An additional similarity lies in their supplementary self-loops in the adjacency matrix inputs. 

However, one key difference between these methods is that the CV authors use raw pairwise votes rather than the derived margins in the construction of their graphs. Put differently, if $a \succ b$ by 100 voters and $b \succ a$ by 25 voters, the CV approach would connect nodes $a$ and $b$ with two edges, $b \rightarrow a$ with weight 100 and $a \rightarrow b$ with weight 25. BallotRank's marginal approach, in contrast, would see a single edge $b \rightarrow a$ with weight 75.

In addition, CV utilizes self-loops with massively larger weights. In BallotRank, prior to column normalization the diagonal sums to one while the off-diagonals are integer vote margins, often in the hundreds or thousands. Under the CV approach, each initial diagonal entry equals the (number of voters) times (the number of candidates less one) minus (the sum of lost votes for that node). In their toy example with three candidates and five million voters, this means subtracting the sum of lost votes (ranging from two to four million per candidate) from 10 million, and adding initial self-loops with weights of six to eight million. Ultimately, this places relatively less emphasis on the observed network structure and relatively more emphasis upon the arbitrarily chosen self-loop weights. 

Finally, although CV differs in its construction of adjacency matrices, it is still the PageRank algorithm that does the heavy lifting -- and this means that the damping parameter is a key factor. However, the authors of \cite{BJNR:ConvergenceVoting} do not provide the analysis for various $d$ and just use $d=1$ (although this is not explicit in the paper). Moreover,  Section 6.2 of \cite{BJNR:ConvergenceVoting} discusses an example in which CV fails to select the Condorcet winner -- even at $d=1$. By contrast, BallotRank guarantees the selection of the Condorcet winner at $d=1$ when one exists. We believe this to be a desirable property, making BallotRank more preferable than CV.

% (see Section 6.2 specifically), we have this: Using $d=0.85$. 

% Their adjacency matrix looks like this:
% \begin{center}
% \begin{tabular}{lrrr}
% \toprule
% Set & A & B & C \\
% \midrule
% A & 12 & 16 & 12 \\
% B  &  4 & 25 & 11 \\
% C     &   8 & 9 & 23 \\
% \bottomrule
% \end{tabular}
% \end{center}
% and yields PageRank scores of A: 0.2022, B: 0.4111, C: 0.3866, failing to pick the Condorcet winner (C).

% Our adjacency matrix looks like this:
% \begin{center}
% \begin{tabular}{lrrr}
% \toprule
% Set & A & B & C \\
% \midrule
% A & 0 & 12 & 4 \\
% B  &  0 & 0.6667 & 2 \\
% C     &   0 & 0 & 0.3333 \\
% \bottomrule
% \end{tabular}
% \end{center}
% and yields BallotRank scores of A: 0.0500, B: 0.1040, C: 0.8460, clearly identifying the Condorcet winner.

% Any other ways BallotRank is different from others?

%%%%%%%%%%%%%%%%%%%%%%%%%%%%%%%%%%%%%%%%%%%%%%%%%%%%%%%%%%%%%%%%%%%%%%%%%%%%%%%%%%%%%%%%

\section{Social choice properties of BallotRank}\label{sec:socialchoice}

In this section, we test BallotRank against some standard criteria from social choice theory. Recall that $\mathcal V$ is the set of voters and  $\mathcal C$ the set of candidates; both sets are  finite. Let
\begin{itemize}
\item $\mathcal P$ be the set of non-empty subsets of $\mathcal C$;
    \item $\mathcal R$ be the set of all complete, transitive, strict orderings on $\mathcal C$; and
    \item $\mathcal W$ be the set of all weak orderings on $\mathcal C$.
\end{itemize}
An element of $\mathcal R^{|\mathcal V|}$ is called a \emph{(preference) profile}. In the setup of elections, a profile is simply the collection of all ranked ballots by all the voters. An element of $\mathcal W$ is an ordering of candidates but ties are allowed.

\begin{defin}\label{d:socialwelfare}
A \emph{social choice function} is a function
$$
\mathcal R^{|\mathcal V|}\longrightarrow \mathcal P.
$$
A \emph{social welfare function} (SWF) is a function
$$
\mathcal R^{|\mathcal V|}\longrightarrow \mathcal W.
$$
\end{defin}
A social choice function should be thought of as an aggregation rule that produces a winner or a subset of winners (thought of as tied for first place). A SWF does that and more, producing a ranking of all candidates with ties allowed.
% \footnote{Ismar: Not sure we need the definition of social choice function. Could just mention it in passing.
BallotRank is clearly a SWF because it naturally produces a ranking of the candidates. All three comparison methods -- minimax, ranked pairs, and Shultze/beatpath -- are usually defined as social choice functions but can be extended to SWFs.\footnote{For minimax, order all candidates according to the margins of their worst pairwise defeat; for ranked pairs, use pairwise victories in descending order of strength to produce a ranking; for Schulze, look at all pairwise beatpaths and order the candidates according to beatpath strength.}

Although most of the social choice criteria we examine are usually stated for social choice functions, they are easily extended to SWFs. Some are discussed in \cite[Section 5]{BJNR:ConvergenceVoting} (Pareto, dictatorship, IIA); we add to this list and consider several others. What we find is that BallotRank does not entirely align with other popular Condorcet completion methods in which criteria it satisfies (Condorcet, anonymity, neutrality, majority, non-dictatorship, Smith, Pareto) or fails (IIA, monotonicity, later-no-harm, no-show, cloning). For example, minimax, ranked pairs, and Schulze satisfy monotonicity while BallotRank does not, and BallotRank satisfies the Condorcet loser and Smith criteria while minimax does not. Given the different outcomes already observed in real-world elections without Condorcet winners, the fact that BallotRank and these other methods satisfy differing social choice criteria ought to come as no surprise.

Below we list and informally define the criteria considered, followed by proofs or counterexamples. 
%Throughout, we assume $d=1$.

%, but will initially ignore the addition of self-loops.  At the end, we will argue that self-loops do not change any of the results or the arguments. 
% \footnote{Ismar: Had this earlier, probably don't need it: We will also assume that the underlying weighted graph is always connected, and will make some remarks about what would change if it is not.}

\begin{itemize}
\item \textbf{Anonymity:} The output of the SWF is invariant under permuted ballots.
\item \textbf{Neutrality:} If the order of two candidates is switched on every ranked ballot, their order is switched in the output of the SWF.
\item \textbf{Non-dictatorship:} There is no voter such that, for all profiles $\mathcal R^{|\mathcal V|}$ and all candidates $a, b$, when this voter prefers $a$ to $b$, the SWF always ranks $a$ above $b$.
\item \textbf{Majority:} If a candidate is ranked first by a majority of voters, then this candidate must be ranked first in the SWF output.
\item \textbf{Condorcet loser:} A Condorcet loser cannot be ranked first by the SWF.
\item \textbf{Pareto:} If all voters prefer $a$ to $b$, then the SWF does not rank $b$ above $a$.
\item \textbf{Smith:} The SWF always ranks every candidate in the Smith set -- the smallest nonempty subset of candidates such that every member defeats every nonmember -- above all those outside it.
\item \textbf{Independence of Irrelevant Alternatives (IIA):} If every voter has the same preferences between $a$ and $b$ across two profiles $\mathbf{R}$ and $\mathbf{R}'$, then the SWF must rank $a$ relative to $b$ identically under both profiles.  
\item \textbf{Monotonicity:} If a profile $\mathbf{R}$ is changed into profile $\mathbf{R}'$ by moving candidate $a$ up in some ranked ballots without changing any other relative orders, then the SWF must not rank $a$ lower in $\mathbf{R}'$ than in $\mathbf{R}$.
\item \textbf{Later-no-harm:} If the SWF declares $a$ to be the winner under profile $\mathbf{R}$, and a voter changes her ballot by ranking additional candidates below $a$, then $a$ should remain the winner.
\item \textbf{No-show/Participation:} If the SWF declares $a$ to be the winner and additional ballots ranking $a$ higher than $b$ are counted, then $b$ should not supplant the winner. 
\item \textbf{Cloning:} If the SWF declares $a$ to be the winner and we introduce one or more $a$ clones -- new candidates ranked adjacent to $a$ on every ballot -- then either $a$ or one of its clones should win.
\end{itemize}

\begin{prop}
    BallotRank satisfies anonymity, neutrality, and non-dictatorship for any $0<d\leq 1$ and majority for $d=1$.
\end{prop}

\begin{proof}
    If two ballots are switched, the pairwise margins do not change. The method hinges upon pairwise margins, hence BallotRank satisfies anonymity.
    
    If candidates $a$ and $b$ are switched on every ballot, that has the effect of swapping the corresponding rows and columns of the matrix $\boldsymbol{\ell}$. This means that the corresponding entries of vector $\mathbf{BR}$ are also swapped; this is precisely what neutrality dictates.
    
    It is a general fact that a SWF which satisfies anonymity must also satisfy non-dictatorship. If a dictator existed, then swapping ballots with that voter would change the winner when the new ballot had a different candidate at the top, but anonymity says this cannot happen. 
    
    It is another general fact that any Condorcet completion method (which BallotRank is for $d=1$) satisfies the majority criterion. Namely, if a candidate is ranked first by a majority of agents, then this candidate beats all other candidates in head-to-head contests and is by definition the Condorcet winner; the majority criterion follows. For $d<1$ the majority criterion is not guaranteed to be satisfied because of teleportation (just like BallotRank is not guaranteed to always find the Condorcet winner when one exists), but this is unlikely to happen in practice. 
\end{proof}

\begin{prop}
    BallotRank satisfies the Condorcet loser criterion for all $0<d\leq 1$.
\end{prop}

\begin{proof}
    If $a$ is the Condorcet loser, then row $a$ of the BallotRank matrix $\boldsymbol{\ell}$ contains only zeros. This means that $a$ is giving all of its weight to other nodes and retains none of it. The only weight it could have is due to teleportation, and the maximum ranking it could obtain in that case is $(1-d)/n$ because the second summand in Equation (\ref{eq:diagweighted}) is always zero. But since $d$ is positive, this number is less than $1/n$. With $n$ candidates, it follows that another candidate's score must be greater than $1/n$,  which in turn means that $a$ cannot be the winner. As a special case, for $d=1$, BallotRank assigns $a$ the value of $0$.  
\end{proof}

\begin{prop}
    BallotRank satisfies the Pareto criterion for all $0<d\leq 1$.
\end{prop}

\begin{proof} 
    We will first argue the case with no self-loops. In that situation, recall that weighted PageRank is given by Equation \eqref{eq:weighteddamped}:$$ \widetilde{PR}(a) = \frac{1 - d}{n} + d \sum_{b\in \mathcal C \, a\neq b} \frac{M(a,b)}{\widetilde{TL}(b)}\widetilde{PR}(b), $$

    where $M(a,b)$, $a\neq b$, is the margin by which $a$ beats $b$ (zero if $a$ does not beat $b$), and $\widetilde{TL}(b)$ is the total loss margin of $b$.

    Consider the columns of the matrix $\boldsymbol{\tilde{\ell}}=(\tilde{\ell}(a,b))$ defined by $\tilde{\ell}(a,b)=\frac{M(a,b)}{\widetilde{TL}(b)}$. By assumption $a$ always beats $b$, so $M(a,b)>0$ and $M(b,a)=0$. Therefore $$ \tilde{\ell}(a,b)=\frac{M(a,b)}{\widetilde{TL}(b)} > 0, \ \ \tilde{\ell}(b,b)=0. $$

    Now fix any candidate $c\notin\{a,b\}$. By assumption, $a$ beats $b$ for every voter. If $b$ beats $c$, then transitivity of individual preferences implies that $a$ also beats $c$. Hence every ballot that contributes to $M(b,c)$ also contributes to $M(a,c)$, and we have $$ M(a,c)\geq M(b,c). $$ 

    Consequently $$ \tilde{\ell}(a,c) = \frac{M(a,c)}{\widetilde{TL}(c)} \geq \frac{M(b,c)}{\widetilde{TL}(c)} = \tilde{\ell}(b,c). $$

    Finally, in column $a$ we have $\tilde{\ell}(b,a)=0$, since unanimity implies that $b$ never beats $a$.

    Putting this together, we see that the row of $\boldsymbol{\tilde{\ell}}$ corresponding to $a$ dominates the row corresponding to $b$ componentwise, with strict inequality in column $b$. That is, $$\tilde{\ell}(a,c) \geq \tilde{\ell}(b,c)\ \  \text{for all $c$ and}\ \ \tilde{\ell}(a,b)>\tilde{\ell}(b,b).$$

    % Multiplying both sides by the nonnegative PageRank vector $\widetilde{\mathbf{PR}}=(\widetilde{PR}(c))$ gives \[ \widetilde{PR}(a) = \frac{1 - d}{n} + d \sum_{c}\tilde{\ell}(a,c)\,\widetilde{PR}(c) \geq \frac{1 - d}{n} + d \sum_{c}\tilde{\ell}(b,c)\,\widetilde{PR}(c) = \widetilde{PR}(b). \] The inequality is strict unless $\widetilde{PR}(b)=0$ (equality could happen if $d=1$ and there is a Condorcet winner which is neither $a$ nor $b$, so BallotRank for both $a$ and $b$ is zero). 
    Multiplying both sides by the nonnegative BallotRank vector $\widetilde{\mathbf{BR}}=(\widetilde{BR}(c))$ gives \[ \widetilde{BR}(a) = \frac{1 - d}{n} + d \sum_{c}\tilde{\ell}(a,c)\,\widetilde{BR}(c) \geq \frac{1 - d}{n} + d \sum_{c}\tilde{\ell}(b,c)\,\widetilde{BR}(c) = \widetilde{BR}(b). \] The inequality is strict unless $\widetilde{BR}(b)=0$ (equality could happen if $d=1$ and there is a Condorcet winner which is neither $a$ nor $b$, so BallotRank for both $a$ and $b$ is zero). 

    Adding self-loops, namely passing from $\boldsymbol{\tilde{\ell}}$ to $\boldsymbol{\ell}$ does not change the argument significantly. Each outcome now keeps some share of its own BallotRank (in proportion to the total margins of the victories it has over others). However, the unanimous preference for $a$ over $b$ still guarantees that $b$ must send a portion of its weight directly to $a$, while $a$ never sends any weight to $b$. Moreover, in every other pairwise contest, the edges going out of $b$ are also weakly dominated by the edges going out of $a$.  
    
    Thus the columns of $\boldsymbol{\ell}$ continue to favor $a$ at least as much as $b$ (and in column $b$ there is still a strict advantage for $a$).  The effect of self-loops is only to alter how much weight each candidate retains overall; the conclusion that $a$ must rank at least as high as $b$ is unchanged.
\end{proof}

\begin{prop}
    BallotRank satisfies the Smith criterion.
\end{prop}

\begin{proof}
    Let $S$ be the Smith set. Then, for any $s\in S$ and $a\in \mathcal C\setminus S$, $s$ beats $a$ and there is therefore a weighted edge from $a$ to $s$. Consequently the weight is distributed from $\mathcal C\setminus S$ to $S$ and so eventually $BR(s)>BR(a)$. This means that all outcomes in the Smith set are ranked above all outcomes in the complement. Adding weighted self-loops does not change this argument since this does not affect the weight flow in one direction, toward the Smith set.
\end{proof}

We devote the remainder of this section to examples of BallotRank failing various criteria. To begin, by Arrow's Theorem \cite{Arrow1951}, we know that BallotRank must fail IIA because it satisfies anonymity, neutrality, and non-dictatorship (and we are assuming unrestricted domain). The following example illustrates.

\begin{example} \textit{BallotRank fails IIA.} 

Recall the four-candidate example in Section \ref{sec:example} where BallotRank ranked $a\succ b\succ d\succ c$. Now consider a new profile where three voters rank $a\succ b\succ c\succ d$, five voters rank $c\succ a\succ b\succ d$, and four voters rank $b\succ c\succ a\succ d$. Note that, as in the original profile, the first eight voters rank $a$ over $b$ and the last four rank $b$ over $a$. The BallotRank matrix is now 
$$ \boldsymbol{\ell}= \begin{pmatrix}
\frac{1}{19} & \frac{96}{103} &  0  &  \frac{1}{3}\\[4pt]
0 & \frac{7}{103} & \frac{16}{19} &  \frac{1}{3}\\[4pt]
\frac{18}{19} & 0 & \frac{3}{19}  &  \frac{1}{3}\\[2pt]
0 & 0 & 0 & 0
\end{pmatrix} $$

Noting the existence of a Condorcet cycle $a\succ b\succ c\succ a$, the margin and BallotRank graphs are now
\begin{center}
\begin{tikzpicture}[node distance=3cm, baseline=(current bounding box.center)]
  % Nodes in a square
  \node[mynode] (a) {$a$};
  \node[mynode, right of=a] (b) {$b$};
  \node[mynode, below of=a] (d) {$d$};
  \node[mynode, below of=b] (c) {$c$};

  % Directed edges with adjusted labels
  \draw[myarrow] (b) -- (a) node[midway, above] {$4$};
  \draw[myarrow] (c) -- (b) node[pos=0.5, right, yshift=4pt] {$2$}; % moved up
  \draw[myarrow] (a) -- (c) node[pos=0.2, right, yshift=4pt] {$6$};
  \draw[myarrow] (d) -- (a) node[pos=0.5, left, yshift=4pt] {$12$}; % moved up
  \draw[myarrow] (d) -- (b) node[pos=0.8, left, yshift=4pt] {$12$};
  \draw[myarrow] (d) -- (c) node[midway, above] {$12$};
\end{tikzpicture}
\ \ \ \ \ \ \ \ 
\begin{tikzpicture}[node distance=3cm, baseline=(current bounding box.center)]
  % Nodes in a square
  \node[mynode] (a) {$a$};
  \node[mynode, right of=a] (b) {$b$};
  \node[mynode, below of=a] (d) {$d$};
  \node[mynode, below of=b] (c) {$c$};

  % Directed edges with adjusted labels
  \draw[myarrow] (b) -- (a) node[midway, above] {$\tfrac{96}{103}$};
  \draw[myarrow] (c) -- (b) node[midway, right] {$\tfrac{16}{19}$};
  \draw[myarrow] (a) -- (c) node[pos=0.2, right, yshift=4pt] {$\tfrac{18}{19}$};    % moved up along arrow
  \draw[myarrow] (d) -- (a) node[midway, left] {$\tfrac{1}{3}$};
  \draw[myarrow] (d) -- (b) node[pos=0.8, left, yshift=4pt] {$\tfrac{1}{3}$}; % moved up along arrow
  \draw[myarrow] (d) -- (c) node[midway, above] {$\tfrac{1}{3}$};

  % Self-loops (rounder, correct orientation)
  \draw[myarrow] (a) edge[loop left,  min distance=10mm, in=150, out=210] node {$\tfrac{1}{19}$} (a);
  \draw[myarrow] (b) edge[loop right, min distance=10mm, in=30, out=330] node {$\tfrac{7}{103}$} (b);
  \draw[myarrow] (c) edge[loop right, min distance=10mm, in=330, out=30] node {$\tfrac{3}{19}$} (c);
  \draw[myarrow] (d) edge[loop left,  min distance=10mm, in=210, out=150] node {$0$} (d);
\end{tikzpicture}
\end{center}

BallotRank now yields
\begin{center}
\begin{tabular}{l|cccc}
\toprule
               & $a$ & $b$ & $c$ &  $d$ \\
\midrule
$d=0.85$ & 0.3086 & 0.3113 & 0.3426 & 0.0375 \\
$d=1.00$ & 0.3183 & 0.3236 & 0.3581 & 0.0000 \\
\bottomrule
\end{tabular}
\end{center}
Even though the relative preferences for $a$ and $b$ remain the same as in the original example, BallotRank has switched the ranking of these two candidates under both damping parameter values, so that the output is now $c\succ b\succ a\succ d$. This illustrates the failure of IIA.\footnote{The second row of rankings also demonstrates that BallotRank resolves Condorcet cycles at $d=1$. Since candidate $d$ is the Condorcet loser, BallotRank appropriately assigns it a value of $0$.}
\end{example}

\begin{example} \textit{BallotRank fails monotonicity.}
    
Suppose the 12 voters in our four-candidate race submitted the following ranked ballots:
%\jfn{Doesn't this example also show a failure of IIA? (Individuals' relative rankings of $b$ and $c$ don't change, but the SWF output does.) Would it be more efficient to use the one example for both?} 
%     \[
% \begin{array}{lcl}
% b\succ a\succ d\succ c, &\quad& b\succ a\succ d\succ c,\\
% b\succ a\succ d\succ c, &\quad& b\succ a\succ d\succ c,\\
% c\succ b\succ a\succ d,      && c\succ b\succ a\succ d,\\
% c\succ b\succ a\succ d,      && c\succ b\succ a\succ d,\\
% d\succ b\succ c\succ a,      && a\succ c\succ d\succ b,\\
% d\succ b\succ c\succ a,      && a\succ c\succ d\succ b,\\
% d\succ c\succ a\succ b,      && d\succ c\succ a\succ b,\\
% d\succ c\succ a\succ b,      && d\succ c\succ a\succ b,\\
% a\succ c\succ d\succ b,      && d\succ b\succ c\succ a,\\
% b\succ a\succ c\succ d,      && b\succ a\succ c\succ d,\\
% c\succ a\succ d\succ b,      && d\succ b\succ a\succ c,\\
% d\succ c\succ b\succ a.      && d\succ c\succ b\succ a.
% \end{array}
% \]

% \begin{tabular}{cc}
% $b\succ a\succ d\succ c$ &  \\
% $b\succ a\succ d\succ c$ &  \\
% $c\succ b\succ a\succ d$      &\\
% $c\succ b\succ a\succ d$      & \\
% $d\succ b\succ c\succ a$     &  $d\succ b\succ a\succ c$\\
% $d\succ b\succ c\succ a$     & \\
% $d\succ c\succ a\succ b$     & \\
% $d\succ c\succ a\succ b$     & \\
% $a\succ c\succ d\succ b$     & \\
% $b\succ a\succ c\succ d$     & \\
% $c\succ a\succ d\succ b$      &  $a\succ c\succ d\succ b$\\
% $d\succ c\succ b\succ a$
% \end{tabular}

\begin{center}
\begin{tabular}{ccccc}
$b\succ a\succ d\succ c$  & & $d\succ b\succ c\succ a$  & & $a\succ c\succ d\succ b$ \\
$b\succ a\succ d\succ c$  & &  $d\succ b\succ c\succ a$ & & $b\succ a\succ c\succ d$ \\
$c\succ b\succ a\succ d$  & & $d\succ c\succ a\succ b$  & & $c\succ a\succ d\succ b$ \\
$c\succ b\succ a\succ d$   & & $d\succ c\succ a\succ b$  & &  $d\succ c\succ b\succ a$\\
\end{tabular}
\end{center}

Now consider another profile in which only two ballots differ, in each case because candidate $a$ is now ranked higher:
\begin{align*}
c\succ a\succ d\succ b\ &\longrightarrow\ a\succ c\succ d\succ b;\\
d\succ b\succ c\succ a\ &\longrightarrow\ d\succ b\succ a\succ c.
\end{align*}
The margin matrices are
\[
\begin{pmatrix}
0 & 0 & 0 & 2\\
4 & 0 & 0 & 0\\
4 & 2 & 0 & 0\\
0 & 2 & 2 & 0
\end{pmatrix}
\ \ \ \ \text{and} \ \ \ \ 
\begin{pmatrix}
0 & 0 & 0 & 2\\
4 & 0 & 0 & 0\\
0 & 2 & 0 & 0\\
0 & 2 & 2 & 0
\end{pmatrix},
\]

while the corresponding BallotRank matrices are

\[
\begin{pmatrix}
\frac{1}{65} & 0 & 0 & \frac{8}{9}\\[4pt]
\frac{32}{65} & \frac{1}{17} & 0 & 0\\[4pt]
\frac{32}{65} & \frac{8}{17} & \frac{3}{19} & 0\\[4pt]
0 & \frac{8}{17} & \frac{16}{19} & \frac{1}{9}
\end{pmatrix}
\ \ \ \ \text{and} \ \ \ \ 
\begin{pmatrix}
\frac{1}{25} & 0 & 0 & \frac{6}{7}\\[4pt]
\frac{24}{25} & \frac{1}{13} & 0 & 0\\[4pt]
0 & \frac{6}{13} & \frac{1}{13} & 0\\[4pt]
0 & \frac{6}{13} & \frac{12}{13} & \frac{1}{7}
\end{pmatrix}.
\]

For the left matrix, BallotRank gives
\begin{center}
\begin{tabular}{l|cccc}
\toprule
               & $a$ & $b$ & $c$ &  $d$ \\
\midrule
$d=0.85$ & 0.2763 & 0.1612 & 0.2513 & 0.3112 \\
%$d=1$ & 0.2851 & 0.1491 & 0.2500 & 0.3158 \\
\bottomrule
\end{tabular}
\end{center}
and for the right it gives
\begin{center}
\begin{tabular}{l|cccc}
\toprule
               & $a$ & $b$ & $c$ &  $d$ \\
\midrule
$d=0.85$ & 0.2677 & 0.2738 & 0.1551 & 0.3034 \\
%$d=1$ & 0.2717 & 0.2826 & 0.1413 & 0.3043 \\
\bottomrule
\end{tabular}
\end{center}
The first profile thus produces the ranking $d\succ a\succ c\succ b$ and the second produces $d\succ b\succ a\succ c$.
So even though $a$ does better in the second profile, BallotRank gives it a worse ranking (and this in fact happens for all $d$). This illustrates the failure of monotonicity. 

Heuristically, this is a consequence of PageRank’s non-locality; the method is defined in terms of the global balancing of flows -- improving one candidate can change the proportional redistribution in ways that indirectly strengthen a competitor. Another way of putting it is that higher-order network effects can overwhelm the small, local boost to $a$.

Note that this example also illustrates the failure of IIA since the relative ranking of $b$ and $c$ did not change from one profile to the other, yet their position in the outputs changed.
\end{example}

\begin{example} \textit{BallotRank fails later-no-harm.}

To see why BallotRank fails the later-no-harm criterion, say there are three candidates, $a$, $b$, and $c$, and four voters with ballots 
$$
a\succ b\succ c, \ \ 
c\succ a\succ b, \ \ 
b\succ c, \ \ 
c\succ a.
$$
(There are two partial ballots here, which are technically not allowed since we are assuming that the domain of our social welfare functions is $\mathcal R^{|\mathcal V|}$, and $\mathcal R$ contains complete orderings. However, it is easy to extend the definition to partial orderings; we will not bother with the details.)

The BallotRank matrix is
\[
\boldsymbol{\ell} = 
\begin{pmatrix}
\frac{1}{3} & \frac{8}{9} & 0\\[4pt]
0 & \frac{1}{9} & \frac{4}{5}\\[4pt]
\frac{2}{3} & 0 & \frac{1}{5}\\[4pt]
\end{pmatrix}
\]
and $a$ wins with the rankings $a=0.387$, $b=0.290$, and $c=0.323$ when the damping parameter is set to one.

But now suppose the voter who had the partial ranking 
$c\succ a$ decides to extend their ballot to $c\succ a\succ b$. The matrix is 
\[
\boldsymbol{\ell'} = 
\begin{pmatrix}
\frac{3}{7} & 1 & 0\\[4pt]
0 & 0 & 0\\[4pt]
\frac{4}{7} & 0 & 1\\[4pt]
\end{pmatrix}
\]
with the rankings $c\succ a\succ b$ (for all values of $d$). Now $c$ is the winner even though the voter only added rankings below $a$.

The intuition is that the added $a\succ b$ comparison increases how much ``flow'' cycles between $b$ and $c$ and how much of that flow avoids returning to $a$. Because BallotRank redistributes by margins with self-loops, seemingly harmless later preferences can re-route centrality in a way that hurts the earlier favorite.
\end{example}

\begin{example} \textit{BallotRank fails no-show.}

Suppose an election produced the following ballots, where the parenthetical number denotes multiplicity:
\begin{align*}
 a \succ b \succ c \ \ (9) \\
b \succ a \succ c \ \ (1) \\
b \succ c \succ a \ \ (5) \\
c \succ a \succ b \ \ (1) \\
c \succ b \succ a \ \ (3)
\end{align*}
Since $a$ beats both $b$ and $c$ by one vote each, it is the Condorcet winner and BallotRank winner as well. (And $b$ beats $c$ by 11 votes.)

Now add two ballots of type $c \succ a \succ b$ (and note that $a$ is ranked above $b$ in both). The margins matrix, with diagonals, and the corresponding BallotRank matrix are
$$
\textbf{M}=   
\begin{pmatrix}
\frac{3}{13} & 3 & 0\\[4pt]
0 & \frac{9}{13} & 9\\[4pt]
1 & 0 & \frac{1}{13}
\end{pmatrix}, 
\ \ \ \ \ \ \ 
\boldsymbol{\ell}=
\begin{pmatrix}
\frac{3}{16} & \frac{13}{16} & 0\\[4pt]
0 & \frac{3}{16} & \frac{117}{118}\\[4pt]
\frac{13}{16} & 0 & \frac{1}{118}
\end{pmatrix}
$$
BallotRank gives $b\succ a\succ c$ (for all $d<1$). Since $b$ is the new winner, this shows the failure of the no-show criterion. 
\end{example} 
   
\begin{example} \textit{BallotRank fails cloning.}

Consider the two profiles presented below, where the right column was obtained from the left column by introducing a clone of candidate $a$. The numbers in parentheses again mean multiplicity:
\begin{align*}
a\succ b\succ c\succ d\ \ (4)  & \qquad 
 a'\succ a\succ b\succ c\succ d\ \ (4) \\
a\succ c\succ b\succ d\ \ (1)  & \qquad 
 a'\succ a\succ c\succ b\succ d\ \ (1)  \\
b\succ c\succ d\succ a\ \ (3)  &  \qquad 
b\succ c\succ d\succ  a'\succ a\ \ (3) \\
c\succ b\succ a\succ d\ \ (2)  &\qquad 
c\succ b\succ  a'\succ a\succ d\ \ (2)  \\
c\succ d\succ a\succ b\ \ (4)  & \qquad 
c\succ d\succ  a'\succ a\succ b\ \ (4)  \\
d\succ a\succ b\succ c\ \ (4)  & \qquad 
d\succ  a'\succ a\succ b\succ c\ \ (4) \\
d\succ b\succ c\succ a\ \ (1)  & \qquad 
d\succ b\succ c\succ  a'\succ a\ \ (1)
\end{align*}

The corresponding BallotRank matrices are
$$
\begin{pmatrix}
\frac{1}{25}&\frac{98}{101}&0&0\\[4pt]
0&\frac{3}{101}&\frac{14}{15}&\frac{28}{285}\\[4pt]
\frac{4}{25}&0&\frac{1}{15}&\frac{252}{285}\\[4pt]
\frac{20}{25}&0&0&\frac{5}{285}
\end{pmatrix}
\ \ \ \ 
\begin{pmatrix}
\frac{7}{1507}&0&\frac{70}{141}&0&0\\[4pt]
\frac{1140}{1507}&\frac{13}{193}&\frac{70}{141}&0&0\\[4pt]
0&0&\frac{1}{141}&\frac{300}{311}&\frac{6}{61}\\[4pt]
\frac{60}{1507}&\frac{30}{193}&0&\frac{11}{311}&\frac{54}{61}\\[4pt]
\frac{300}{1507}&\frac{150}{193}&0&0&\frac{1}{61}
\end{pmatrix}.
$$

For the first, BallotRank ($d=1$) gives $a$ as the winner ($a=0.2669 \succ b=0.2641$), while for the second, the winner is $b$ ($b=0.2369 \succ a'=0.2220 \succ a=0.1182$). This shows BallotRank's failure of the cloning criterion.
\end{example} 

\begin{table}
\begin{center}
\begin{tabular}{r|ccccccc}
    \toprule
    Method & Anonymity & Neutrality & Majority & ND & Pareto & C Loser & Smith \\
    \midrule
    BallotRank   & \checkmark & \checkmark & \checkmark & \checkmark & \checkmark & \checkmark & \checkmark \\
    Minimax      & \checkmark & \checkmark & \checkmark & \checkmark & \checkmark & \xmark & \xmark \\
    Ranked Pairs & \checkmark & \checkmark & \checkmark & \checkmark & \checkmark & \checkmark & \checkmark \\
    Schulze      & \checkmark & \checkmark & \checkmark & \checkmark & \checkmark & \checkmark & \checkmark \\
    \bottomrule
\end{tabular}
\caption{Most comparison methods pass the criteria that BallotRank passes.}
\label{tab:passes}
\end{center}
\end{table}

\begin{table}
\begin{center}
\begin{tabular}{r|ccccc}
    \toprule
    Method & Clones & Monotone & No-Show & LNH & IIA \\
    \midrule
    BallotRank   & \xmark     & \xmark     & \xmark & \xmark & \xmark \\
    Minimax      & \xmark     & \checkmark & \xmark & \xmark & \xmark \\
    Ranked Pairs & \checkmark & \checkmark & \xmark & \xmark & \xmark \\
    Schulze      & \checkmark & \checkmark & \xmark & \xmark & \xmark \\
    \bottomrule
\end{tabular}
\caption{Failing cloning and monotonicity sets BallotRank apart.}
\label{tab:fails}
\end{center}
\end{table}

Tables \ref{tab:passes} and \ref{tab:fails} succinctly summarize BallotRank's performance along these social choice criteria and contrast it with minimax, ranked pairs, and Schulze/beatpath.

\section{Conclusions and further directions}\label{sec:conclusion}

Building upon insights from \cite{mcgarvey1953theorem} and \cite{bp1998cn}, and sharing aspects of \cite{BJNR:ConvergenceVoting}, we introduced BallotRank, a Condorcet completion method. Like McGarvey (and his advisor Kenneth May), we harness an equivalence between preference profiles and directed graphs. Unlike McGarvey, however, we leverage additional information in the form of pairwise margins. Like PageRank, BallotRank can utilize a damping parameter $d$ to rank candidates in the presence of Condorcet winners and losers. Unlike PageRank, however, BallotRank with its self-loops manages to make finer distinctions amongst the members of Condorcet cycles. Like convergence voting, BallotRank recognizes the problem these cycles present to the task of ranking candidates and addresses it via self-loops. Unlike convergence voting, however, BallotRank constructs low-weight self-loops from pairwise margins, guaranteeing that it identifies the Condorcet winner when $d=1$ and effectively ensuring it at other suitably high values of $d<1$.\footnote{At least, we have yet to find a real-world election where it failed to identify the Condorcet winner at $0.5 \leq d<1$.} Like other notable Condorcet completion methods such as minimax, ranked pairs, and Schulze/beatpath, BallotRank identifies Condorcet winners at $d=1$. Unlike these methods, however, BallotRank is a natural social welfare function offering a full ranking of candidates on the basis of higher-order relationships between preferences.

This paper introduced BallotRank intuitively and rigorously, before demonstrating its application on a sample of tens of thousands of online polls and real-life legislative elections. We find that it correctly identified the Condorcet winner wherever one existed and, when one did not, frequently selected a different winner from comparison methods like minimax, ranked pairs, and Schulze/beatpath. We also prove that BallotRank satisfies common social choice criteria, including anonymity, neutrality, non-dictatorship, majority, Condorcet loser, Pareto, and Smith, while failing IIA, monotonicity, later-no-harm, no-show, and cloning. 

From a practical perspective, we have several recommendations. If the goal is to identify a single winner, then simply implement BallotRank using the power method and a damping parameter value of $d=1$. This will ensure that a Condorcet winner, should one exist, will be identified. When it does not, BallotRank will produce a full ranking and select a reasonable winner. If instead the goal is to produce a complete ranking of the candidates (or select multiple winners), we suggest using the power method to run BallotRank with $d \in [0.85, 1)$. While this will not guarantee successful identification of a Condorcet winner, it will almost invariably do so anyway while producing a full ranking as desired.

Several promising directions for future work emerge from the BallotRank framework. Two natural avenues for exploration are alternative ways of encoding information from the preference profiles into graphs, and alternative algorithms for converting these graphs into social preferences. Along the first line, for example, Atkinson et al. \cite{AGHO:StrongMaxCirc} and Bana et al. \cite{BJNR:ConvergenceVoting} represent profiles as bidirectional edges weighted by pairwise preferences. Along the second line, 
% One natural extension is to explore alternative ways of encoding information from the preference graph beyond the PageRank-style flow used here. While BallotRank relies on a column-stochastic normalization and self-loops whose weights are derived from total win margins, 
exploring other centrality measures (e.g., eigenvector centrality) or random walk–based scores may yield distinct and interpretable social welfare functions. Understanding which features of BallotRank are essential for Condorcet consistency and which can be modified without sacrificing desirable properties could clarify the space of admissible graph-based voting rules. 

A second set of questions concerns robustness and structural sensitivity. Empirically, BallotRank differs from standard Condorcet completion methods in a subset of elections we tested, most notably in the Scottish local council data (see Table \ref{tab:no-condorcet}), raising the question of what graph-theoretic features drive these divergences. Is disagreement correlated with cycle size, margin asymmetry, degree heterogeneity, or the presence of strong out-of-cycle candidates that act as conduits for redistributed mass? Similarly, in elections featuring Condorcet cycles, BallotRank can exhibit bifurcations as the damping parameter $d$ varies, producing different winners at different values of $d$, as was the case in four of the Scottish elections we studied. As rare as this phenomenon seems to be, understanding it more systematically, possibly by fixing $d$ and constructing extremal examples or by characterizing thresholds at which the stationary distribution shifts qualitatively, could shed light on the method’s sensitivity to global versus local structure. Such analysis may suggest principled, or even adaptive and data-driven, guidelines for choosing $d$ in practice.

\bibliographystyle{alpha}
%\bibliographystyle{plain}
%
%\begin{thebibliography}{WRWX23}

\bibliography{all-the-cites}
%\bibliography{../SocialChoiceBib}

%\end{thebibliography}

\end{document}